\documentclass[11pt,letter]{article}

\usepackage{amsmath,amsfonts}
\usepackage{graphicx}
\usepackage{amssymb}
\usepackage{hyperref}
\usepackage{bbm}
\usepackage{microtype}
\usepackage{xcolor}
\usepackage{enumerate}
\usepackage{enumitem}
\usepackage{wrapfig}
\usepackage{fullpage}

\usepackage{amsthm}
\theoremstyle{plain}
\newtheorem{theorem}{Theorem}
\newtheorem{lemma}[theorem]{Lemma}

\newtheorem{claim}{Claim}
\newtheorem{observation}[theorem]{Observation}

\graphicspath{{./lipics/}{./figures/}}

\newcommand{\mypar}[1]{\smallskip\noindent\textbf{#1}}

\newcommand{\ie}{{i.e.}}

\newcommand{\ZZ}{\mathbb{Z}}
\newcommand{\slope}{\mathrm{slope}}
\newcommand{\conv}{\mathrm{conv}}
\newcommand{\proj}{\mathrm{proj}}

\newcommand{\Fn}{F}
\newcommand{\Gn}{G}
\newcommand{\FSn}{\bar{F}}
\newcommand{\GSn}{\bar{G}}

\newcommand{\later}[1]{{}}
\newcommand{\old}[1]{{}}
\long\def\ignore#1{}

\usepackage{tikz}
\newcommand{\chain}{\text{\begin{tikzpicture}\draw[->] (0,-1ex) +(180:1.3ex) arc(180:90:1.3ex);\end{tikzpicture}}}
\newcommand{\chainTopLeft}{\chain}
\newcommand{\chainBottomLeft}{\rotatebox[origin=c]{90}{\chain}}
\newcommand{\chainBottomRight}{\rotatebox[origin=c]{180}{\chain}}
\newcommand{\chainTopRight}{\rotatebox[origin=c]{270}{\chain}~\!}

\newcommand{\capTop}{\text{\ensuremath{\curvearrowright}}}
\newcommand{\capLeft}{\rotatebox[origin=c]{90}{\capTop}}
\newcommand{\capBottom}{\rotatebox[origin=c]{180}{\capTop}}
\newcommand{\capRight}{\rotatebox[origin=c]{270}{\capTop}}

\newcommand{\Cross}{\prod}

\hypersetup{
hidelinks,
colorlinks=true,
citecolor=[rgb]{0.121 0.47 0.705},
linkcolor=[rgb]{0.7 0 0.4},
urlcolor=[rgb]{0.121 0.47 0.705}
}

\title{Convex Polygons in Cartesian Products\thanks{This work was initiated at the 2017 Fields Workshop on Discrete and Computational Geometry (Carleton University, Ottawa, ON, July 31--August 4, 2017).}}

\author{%
Jean-Lou De Carufel\,%
\thanks{School of Electrical Engineering and Computer Science, University of Ottawa, Canada, 
\texttt{jdecaruf@uottawa.ca}. This work has been supported by NSERC.}
\and
Adrian Dumitrescu\,%
\thanks{Algoresearch L.L.C., Milwaukee, WI, USA, 
\texttt{ad.dumitrescu@gmail.com}}
\and
Wouter Meulemans\,%
\thanks{Department of Mathematics and Computer Science, TU Eindhoven, The Netherlands, 
\texttt{w.meulemans@tue.nl}}
\and 
Tim Ophelders\,\footnotemark[4]
\thanks{Department of Information and Computing Sciences, Utrecht University, The Netherlands, 
\texttt{t.a.e.ophelders@uu.nl}}
\and
Claire Pennarun\,%
\thanks{LIRMM, CNRS \& Universit\'e de Montpellier, France,
\texttt{claire.pennarun@gmail.com}}
\and
Csaba D.\ T\'oth\,%
\thanks{Department of Mathematics, California State University Northridge, Los Angeles, CA; and Department of Computer Science, Tufts University, Medford, MA, USA,
\texttt{csaba.toth@csun.edu}}
\and
and Sander Verdonschot\,%
\thanks{School of Computer Science, Carleton University, Canada, \texttt{sander@cg.scs.carleton.ca}}
}

\date{}

\begin{document}

\maketitle

\begin{abstract}
We study several problems concerning convex polygons whose vertices lie in a Cartesian product of two sets of $n$ real numbers (for short, \emph{grid}). First, we prove that every such grid contains $\Omega(\log n)$ points in convex position and that this bound is tight up to a constant factor. We generalize this result to $d$ dimensions (for a fixed $d\in \mathbb{N}$), and obtain a tight lower bound of $\Omega(\log^{d-1}n)$ for the maximum number of points in convex position in a $d$-dimensional grid.
Second, we present polynomial-time algorithms for computing the longest $x$- or $y$-monotone convex polygonal chain in a grid that contains no two points with the same $x$- or $y$-coordinate. We show that the maximum size of a convex polygon with such unique coordinates can be efficiently approximated up to a factor of $2$. Finally, we present exponential bounds on the maximum number of point sets in convex position in such grids, and for some restricted variants. These bounds are tight up to polynomial factors.
\end{abstract}

\section{Introduction}
\label{sec:intro}

Can a convex polygon $P$ in the plane be reconstructed from the projections of its vertices to the coordinate axes?
Assuming that no two vertices of $P$ share the same $x$- or $y$-coordinate, we arrive at the following problem: given two sets, $X$ and $Y$, each containing $n$ real numbers, does the Cartesian product $X \times Y$ support a convex polygon with $n$ vertices?
We say that $X\times Y$ \emph{contains} a polygon $P$ if every vertex of $P$ is in $X\times Y$; and $X\times Y$ \emph{supports} $P$ if it contains $P$ and no two vertices of $P$ share an $x$- or $y$-coordinate. For short, we call the Cartesian product $X\times Y$ an \emph{$n\times n$ grid}.

Not every $n \times n$ grid supports a convex $n$-gon. This is the case already for $n = 5$ (see Figure~\ref{fig:smallN}). Several interesting questions arise: can we decide efficiently whether an $n\times n$-grid supports a convex $n$-gon? How can we find the largest $k$ such that it contains (resp., supports) a convex $k$-gon? What is the largest $k$ such that \emph{every} $n \times n$ grid supports a convex $k$-gon? How many convex polygons does an $n\times n$ grid contain, or support? We initiate the study of these questions for convex polygons, and their higher dimensional variants for convex polyhedra.

\mypar{Our results.}
We first show that every $n \times n$ grid contains (resp., supports) a convex polygon with $(1-o(1))\log n$ vertices\footnote{All logarithms in this paper are of base 2.}; this bound is tight up to a constant factor: there are $n \times n$ grids that do not contain convex polygons with more than $4(\lceil \log n \rceil +1)$ vertices. We generalize our upper and lower bounds to higher dimensions, and show that every $d$-dimensional Cartesian product $\Cross_{i=1}^d X_i$, where $|X_i|=n$ and $d$ is constant, contains $\Omega(\log^{d-1}n)$ points in convex position; this bound is also tight apart from constant factors (Section~\ref{sec:bounds}).
Next, we present polynomial-time algorithms to find a maximum supported convex polygon that is $x$- or $y$-monotone.
We show how to efficiently approximate the maximum size of a supported convex polygon up to a factor of two (Section~\ref{sec:algorithms}). Finally, we present tight asymptotic bounds for the maximum number of convex polygons supported by an $n \times n$ grid (Section~\ref{sec:counting}). We conclude with open problems (Section~\ref{sec:con}).

\mypar{Related work.}
Erd\H{o}s and Szekeres proved, as one of the first Ramsey-type results in combinatorial geometry~\cite{Chapter11}, that for every $k\in \mathbb{N}$, a sufficiently large point set in the plane in general position contains $k$ points in convex position. The minimum cardinality of a point set that guarantees $k$ points in convex position is known as the Erd\H{o}s--Szekeres number, $f(k)$. They proved that $2^{k-2}+1\leq f(k)\leq \binom{2k-4}{k-2}+1=4^{k(1-o(1))}$, and conjectured that the lower bound is tight~\cite{erdos1960}.
The current best upper bound, due to Suk~\cite{Suk1017}, is $f(k)\leq 2^{k(1+o(1))}$.
In other words, every set of $n$ points in general position in the plane contains $(1-o(1))\log{n}$ points in convex position, and this bound is tight up to lower-order terms.

In dimension $d\geq 3$, the asymptotic growth rate of the Erd\H{o}s--Szekeres number is not known. By the Erd\H{o}s--Szekeres theorem, every set of $n$ points in general position in $\mathbb{R}^d$ contains $\Omega(\log n)$ points in convex position (it is enough to find points whose projections onto a generic plane are in convex position). For every constant $d\geq 2$, K{\'a}rolyi and Valtr~\cite{KV03} and Valtr~\cite{Valtr92} constructed  $n$-element sets in general position in $\mathbb{R}^d$ in which no more than $O(\log^{d-1}n)$ points are in convex position. Both constructions are recursive, and one of them is related to high-dimensional Horton sets~\cite{Valtr92}. These bounds are conjectured to be optimal apart from constant factors. Our results establish the same $O(\log^{d-1}n)$ upper bound for Cartesian products, for which it is tight apart from constant factors. However, our results do not improve the bounds for points in general position.

Algorithmically, one can find a largest convex cap in a given set of $n$ points in $\mathbb{R}^2$ in $O(n^2\log n)$ time by dynamic programming~\cite{EdelsbrunnerG89}, and a largest subset in convex position in $O(n^3)$ time~\cite{chvatal1980finding,EdelsbrunnerG89}.
The same approach can be used for counting the number of convex polygons contained in a given point set~\cite{MRSW95}. While this approach applies to grids, it is unclear how to include the restriction that each coordinate is used at most once. On the negative side, finding a largest subset in convex position in a point set in $\mathbb{R}^d$ for dimensions $d\geq 3$ was recently shown to be  NP-hard~\cite{giannopoulos2013computational}.

There has been significant interest in \emph{counting} the number of convex polygons in various point sets.
Answering a question of Hammer, Erd\H{o}s~\cite{erdos1978} proved that every set of $n$ points in general position in $\mathbb{R}^2$ contains $\exp(\Theta(\log^2 n))$ subsets in convex position, and this bound is the best possible.
B\'ar\'any and Pach~\cite{barany1992number1}
showed that the number of convex polygons in an $n \times n$ section of the integer lattice is $\exp\left( O(n^{1/3}) \right)$. B\'ar\'any and Vershik~\cite{barany1992number2} generalized this bound to $d$-dimensions and showed that there are  $\exp\left( O(n^{(d-1)/(d+1)}) \right)$ convex polytopes in an $n \times \cdots \times n$ section of $\ZZ^d$. Note that the exponent is sublinear in $n$ for every $d\geq 2$. We prove that an $n \times n$ Cartesian product can contain $\exp(\Theta(n))$ convex polygons, significantly more than integer grids, and our bounds are tight up to polynomial factors.

Motivated by integer programming and geometric number theory, lattice polytopes (whose vertices are in $\mathbb{Z}^d$) have been intensely studied; refer to~\cite{Bar08,Chapter7}. However, results for lattices do not extend to arbitrary Cartesian products. Recently, several deep results have been established for Cartesian products in incidence geometry and additive combinatorics~\cite{RSS17,RSS16,RSZ16,SSW13}, while the analogous statements
for point sets in general position remain elusive.

\mypar{Definitions.}
A polygon $P$ in $\mathbb{R}^2$ is \emph{convex} if all of its internal angles are strictly smaller than $\pi$.
A point set in $\mathbb{R}^2$ is in \emph{convex position} if it is the vertex set of a convex polygon; and it is in \emph{general position} if no three points are collinear. Similarly, a polyhedron $P$ in $\mathbb{R}^d$ is \emph{convex} if it is the convex hull of a finite set of points. A point set in $\mathbb{R}^d$ is in \emph{convex position} if it is the vertex set of a convex polyhedron; and it is in \emph{general position} if no $d+1$ points lie on a hyperplane. The \emph{1-skeleton} of a polyhedron is the graph formed by its vertices and edges. In $\mathbb{R}^d$, we say that the $x_d$-axis is \emph{vertical}, hyperplanes orthogonal to the $x_d$-axis are \emph{horizontal}, and  the above-below relationship is understood with respect to $x_d$ coordinates.

We consider special types of convex polygons in the plane. Let $P$ be a convex polygon with vertices $((x_1,y_1),\ldots,(x_k,y_k))$ in clockwise order. We say that $P$ is a \emph{convex cap} if the $x$- or $y$-coordinates are strictly monotone, and a \emph{convex chain} if both the $x$- and $y$-coordinates are strictly monotonic. We distinguish four types of convex caps (resp., chains) based on the monotonicity
of the coordinates as follows:

\begin{itemize}
	\item \emph{convex caps} come in four types $\{\capTop,\capLeft,\capBottom,\capRight\}$. We have\\
	\begin{tabular}{ll}
		$P\in\capTop$    & if and only if~$(x_i)_{i=1}^k$ strictly increases;\\
		$P\in\capLeft$   & if and only if~$(y_i)_{i=1}^k$ strictly increases;\\
		$P\in\capBottom$ & if and only if~$(x_i)_{i=1}^k$ strictly decreases;\\
		$P\in\capRight$  & if and only if~$(y_i)_{i=1}^k$ strictly decreases;
	\end{tabular}
	\item \emph{convex chains} come in four types~$\{\chainTopLeft,\chainTopRight,\chainBottomRight,\chainBottomLeft\}$. We have\\
	\begin{tabular}{llll}
	$\chainTopLeft=\capLeft\cap\capTop$, & $\chainTopRight=\capTop\cap\capRight$, & $\chainBottomLeft=\capBottom\cap\capLeft$, & $\chainBottomRight=\capRight\cap\capBottom$.
	\end{tabular}
\end{itemize}

\mypar{Initial observations.}
It is easy to see that for $n\in \{3,4\}$, every $n\times n$ grid supports a convex $n$-gon. However, there exists a $5\times 5$ grid that does not support any convex pentagon (cf.~Figure~\ref{fig:smallN}).
Interestingly, every $6\times 6$ grid supports a convex pentagon.

\begin{figure}[htb]
	\centering
	\includegraphics{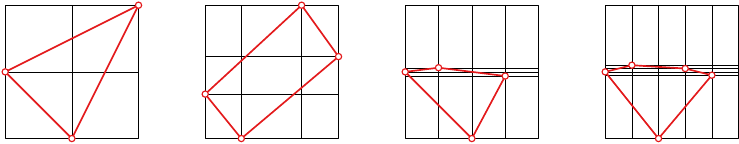}%
	\caption{Maximum-size supported convex polygons of respective sizes~$3$, $4$, $4$, and~$5$ in $n\times n$ grids, where $n$ is between~$3$ and~$6$.}%
	\label{fig:smallN}
\end{figure}

\begin{lemma}
	Every $6 \times 6$ grid $X \times Y$ supports a convex polygon of size at least~$5$.%
	\label{lem:lowerSix}
\end{lemma}
\begin{proof}
    Refer to Figure~\ref{fig:lem-lowerSix} for illustration. Let $x_1, \ldots, x_6$ and $y_1, \ldots, y_6$ denote the vertices of $X$ and $Y$ in increasing order.
	Consider first the $4 \times 4$ subgrid obtained by omitting the extremal values, that is, $\{ x_2, \ldots, x_5 \} \times \{ y_2, \ldots, y_5 \}$.
	As the $2\times 2$ grid $\{x_3,x_4\}\times \{y_3,y_4\}$ consists of four corners of a rectangle, at least two of these points are not collinear with $(x_2,y_2)$ and $(x_5,y_5)$; let $(x',y')$ denote such a point and $x''\in \{x_3,x_4\}$ and $y''\in\{y_3,y_4\}$ denote the unused coordinates. Now, $P' = \{(x_2,y_2),(x',y'),(x_5,y_5)\}$ is a supported convex chain of the subgrid and thus of $X \times Y$. By construction, $P'$ is either in $\chainTopLeft$ or $\chainBottomRight$; without loss of generality, assume the former.
	
\begin{figure}[htb]
\centering
\includegraphics{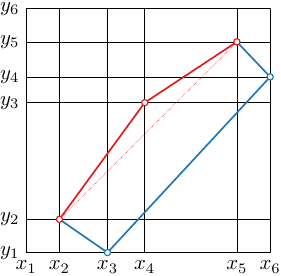}
\caption{A $6\times 6$ grid, in which the convex chain $P'$ (red) has size 3, and it is extended to a convex pentagon $P$ (red and blue) supported by the grid.}
\label{fig:lem-lowerSix}
\end{figure}

	We now extend $P'$ to a convex polygon $P$ of size~$5$ by appending two vertices: $(x_6,y'')$ and~$(x'',y_1)$.
	By construction, all five vertices of $P$ have distinct coordinates.
	Moreover, the points in $P$ are in convex position, as each is an extreme point in $P$:
	$(x_2,y_2)$ has the lowest $x$-coordinate; $(x_5,y_5)$ has the highest $y$-coordinate; $(x_6,y'')$ has the maximum $x$-coordinate; and $(x'',y_1)$ has the minimum $y$-coordinate. Finally, $(x',y')$ is an extreme point in $P'$, hence $P'$ has a supporting halfplane $h$ whose boundary contains $(x',y')$. Since $P'\subset h$, and we have $x_2<x''<x_4$ and $y_2<y''<y_4$, the halfplane $h$ contains $(x_6,y'')$ and $(x'',y_1)$, as well, $(x',y')$ is an extreme point in $P$. Thus, we conclude that $P$ is a supported convex polygon of size $5$ on $X \times Y$.
\end{proof}

Note that, although every pentagon supported by a grid lies on a $5 \times 5$ subgrid, the proof above does not work on a $5 \times 5$ grid. If we start with a $5\times 5$ grid, and choose $P'$ from an off-center $4\times 4$ subgrid, then we might not add be able to add two points from the remaining $2\times 2$ subgrid. Specifically, the coordinates of the $4\times t$ grid may force $P'$ to be of a specific type: the example in Figure~\ref{fig:lem-lowerSix} prevents $P'$ from being in $\chainBottomRight$ as all four central points are above the dotted line. Starting with a $6\times6 $ grid, we know there are two extra points on the outer boundary that can be added to $P'$, regardless of the type of the chain $P'$.

\section{Extremal bounds for convex polytopes in Cartesian products}
\label{sec:bounds}

\subsection{Lower bounds in the plane}
\label{ssec:lowerbound2D}
In this section, we show that for every $n\geq 3$, every $n\times n$ grid supports a convex polygon with $\Omega(\log n)$ vertices. The results on the Erd\H{o}s--Szekeres number cannot be used directly, since they crucially use the assumption that the given set of points is in general position. An $n \times n$ section of the integer lattice is known to contain $\Theta(n)$ points in general position~\cite{erdos1951}; the maximum number of points in general position is conjectured to be $\frac{\pi}{\sqrt{3}}n(1+o(1))$~\cite{GK68,Pegg05}.
However, this result does not apply to arbitrary Cartesian products.
It is worth noting that higher dimensional variants
for the integer lattice are poorly understood: it is known that an $n\times n\times n$ section of $\mathbb{Z}^3$ contains $\Theta(n^2)$ points no three of which are collinear~\cite{PW2007}, but no similar statements are known in
higher dimensions. We use a recent result from incidence geometry.

\begin{lemma}[Payne and Wood~\cite{payne2013general}]
\label{lem:GenPos}
Every set of $N$ points in the plane with at most~$\ell$ collinear, where $\ell\leq O(\sqrt{N})$, contains a set of~$\Omega(\sqrt{N/ \log \ell})$ points in general position.
\end{lemma}

\begin{lemma}
Every $n \times n$ grid supports a convex polygon of size $(1-o(1))\log n$.
\label{lem:lowerGeneral}
\end{lemma}
\begin{proof}
Every $n\times n$ grid contains a set of~$\Omega(\sqrt{n^2 / \log n})=\Omega(n/\sqrt{\log n})$
points in general position by applying Lemma~\ref{lem:GenPos} with $N=n^2$ and $\ell=n$.
Discarding points with the same $x$- or $y$-coordinate reduces the size by a factor at most $\frac{1}{4}$, so this asymptotic bound also holds when coordinates in $X$ and $Y$ are used at most once. By Suk's result~\cite{Suk1017},
this set of $\Omega(n/\sqrt{\log n})$ points contains a convex polygon with at least
$(1-o(1))\big(\log(n / \sqrt{\log n})\big) = (1-o(1))\log n$ vertices,
which have distinct $x$- and $y$-coordinates by construction, as required.
\end{proof}

\subsection{Upper bounds in the plane}
\label{ssec:upperbound2D}

For the upper bound, we construct $n\times n$ Cartesian products that do not support large convex chains. For~$n=8$, such a grid is depicted in Figure~\ref{fig:upperbound}.

\begin{figure}[htbp]\centering
	\includegraphics{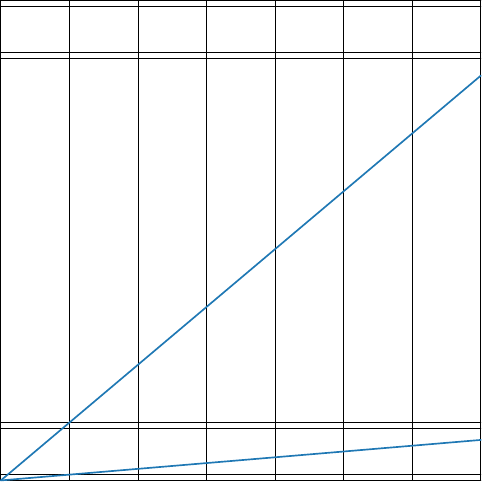}%
	\caption{An $8 \times 8$ grid without convex chains of size greater than~$4=\log 8+1$, where $X=\{0,1,\ldots,7\}$ and $Y=\{0,1,16,17,256, 257,272,273\}$. Two lines through pairs of grid points are drawn in blue.}%
	\label{fig:upperbound}
\end{figure}

\begin{lemma}
  For every $n\in \mathbb{N}$, there exists an $n \times n$ grid that
  contains at most~$\lceil \log n \rceil +1$ points in a convex chain;
  consequently, at most~$4(\lceil \log n \rceil +1)$ points in convex position.
\label{lem:upper2}
\end{lemma}
\begin{proof}
    Let~$g(n)$ be the maximum integer such that for all $n$-element sets $X,Y\subset \mathbb{R}$,
    the grid~$X\times Y$ supports a convex polygon of size~$g(n)$; clearly~$g(n)$ is nondecreasing. Let~$k$ be the minimum integer such that~$n \leq 2^k$; thus $\lceil \log n \rceil \leq k$ and $g(n) \leq g(2^k)$. We show that~$g(2^k) \leq 4(k+1)$ and thereby establish that~$g(n) \leq 4(k+1)$.

    Assume, without loss of generality, that~$n=2^k$, and let~$X=\{0,\dots,n-1\}$.
    For a~$k$-bit integer~$m$, let~$m_i$ be the bit at its~$i$-th position, such that~$m=\sum_{i=0}^{k-1}m_i 2^i$.
    Let~$Y=\{ \sum_{i=0}^{k-1} m_i (2n)^{i} : 0\leq m \leq n-1 \}$ (see Figure~\ref{fig:upperbound}).
    Both~$X$ and~$Y$ are symmetric:~$X=\{\max(X)-x: x\in X\}$ and~$Y=\{\max(Y)-y:
y\in Y\}$. Thus, it suffices to show that no convex chain $P\in\chainTopLeft$ of size greater than~$k+1$ exists.

    Consider two points, $p=(x,y)$ and~$p'=(x',y')$, in $X\times Y$ such that $x<x'$ and $y<y'$.
    Assume $y=\sum_{i=0}^{k-1} m_i (2n)^{i}$ and~$y'=\sum_{i=0}^{k-1} m'_i (2n)^{i}$.
    The slope of the line spanned by $p$ and~$p'$ is~$\slope(p,p')=\sum_{i=0}^{k-1} (m'_i-m_i) (2n)^{i} / (x'-x)$.
    Let~$j$ be the largest index such that~$m_j\neq m'_j$. Then $y<y'$ implies $m_j<m'_j$,
        and we can bound the slope as follows:
    \[
    \slope(p,p')
    \geq \frac{(2n)^{j}-\sum_{i=0}^{j-1}(2n)^{i}}{x'-x}
    > \frac{(2n)^{j}-2(2n)^{j-1}}{n-1}
    = 2\cdot (2n)^{j-1},
    \]

    \[
    \slope(p,p')
    \leq \frac{\sum_{i=0}^j (2n)^{i}}{x'-x}
    \leq \frac{\sum_{i=0}^j (2n)^{i}}{1}
    =\frac{(2n)^{j+1}-1}{2n-1}
    < 2\cdot (2n)^{j}\,.\]
    Hence, $\slope(p,p')\in I_j$, where $I_j=(2\cdot (2n)^{j-1},2\cdot(2n)^{j})$.
    Let us define the family of intervals~$I_0,I_1,\dots,I_{k-1}$ analogously, and
    note that these intervals are pairwise disjoint.
    Suppose that some convex chain~$P\in\chainTopLeft$ contains more than~$k+1$ points.
    Since the slopes of the first~$k+1$ edges of~$P$ decrease monotonically, by the pigeonhole principle, there must be three consecutive vertices~$p=(x,y)$, $p'=(x',y')$, and $p''=(x'',y'')$ of~$P$ such that both~$\slope(p,p')$ and~$\slope(p',p'')$ are in the same interval, say $I_j$.
    Assume that $y=\sum_{i=0}^{k-1} m_i (2n)^{i}$, $y'=\sum_{i=0}^{k-1} m'_i (2n)^{i}$,
        and $y''=\sum_{i=0}^{k-1} m''_i (2n)^{i}$.
    Then~$j$ is the largest index such that~$m_j\neq m'_j$, and also the largest index such that~$m'_j\neq m''_j$.
    Because~$m<m'<m''$, we have~$m_j<m'_j<m''_j$, which is impossible since each of~$m_j$,~$m'_j$ and~$m''_j$ is either~$0$ or~$1$.

Hence, $X\times Y$ does not contain any convex chain in $\chainTopLeft$ of size greater than~$k+1$.
Analogously, every convex chain in $\chainTopRight$, $\chainBottomRight$, or
$\chainBottomLeft$ has at most $k+1$ vertices.
Consequently, $X\times Y$ contains at most $4(k+1)$ points in convex position.
\end{proof}

\subsection{Upper bounds in higher dimensions}
\label{ssec:UpperBoundHighDim}

We construct Cartesian products in $\mathbb{R}^d$, for $d\geq 3$, that match the best known upper bound $O(\log^{d-1}n)$ for the Erd\H{o}s--Szekeres numbers in $d$-dimensions for points in general position. Our construction generalizes the ideas from the proof of Lemma~\ref{lem:upper2} to $d$-space.

\begin{lemma}
Let $d\geq 2$ be an integer. For every integer $n\geq 2$, there exist $n$-element sets $X_i\subseteq \mathbb{R}$ for $i=1, \ldots, d$, such that the Cartesian product $X=\Cross_{i=1}^d X_i$ contains at most $O(\log^{d-1} n)$ points in convex position.
\label{lem:upperd}
\end{lemma}

Before proving Lemma~\ref{lem:upperd}, we introduce some additional terminology
and give a brief overview of key ideas. Let $d$ be a positive integer.  Let $\mathbf{e}_d$ be the standard basis vector parallel to the $x_d$-axis in $\mathbb{R}^d$.
For a point $a\in \mathbb{R}^d$, let $a^\proj$ denote the orthogonal projection of $a$ to the horizontal hyperplane $x_d=0$. Let $P\subset \mathbb{R}^d$ be a finite set in convex position; refer to Figure~\ref{fig:silhouette}.
The point set $P$ is \emph{full-dimensional} if no hyperplane contains $P$.
The orthogonal projection of $\conv(P)$ to the hyperplane $x_d=0$ is a convex polytope in $\mathbb{R}^{d-1}$ that we denote by $\conv(P)^\proj$. The \emph{silhouette} of $P$ is the subset of points in $P$ whose orthogonal projection to $x_d=0$ lies on the boundary of $\conv(P)^\proj$. Note that for every point $p\in P$ in the silhouette of $P$, the projection $p^\proj$ is either a vertex of $\conv(P)^\proj$ or lies in the relative interior of a face of $\conv(P)^\proj$.
Since no three points in $P$ are collinear, at most two points in $P$ are projected to any point in the hyperplane $x_d=0$.
A point $p\in P$ is an \emph{upper} (resp., \emph{lower}) vertex if $P$ lies in the closed halfspace below (resp., above) some tangent hyperplane of $\conv(P)$ at $p$ (a point in $P$ may be both upper and lower vertex).
For an integer $k\in \{1,\ldots , d\}$, a $k$-dimensional flat (for short, \emph{$k$-flat}) is \emph{axis-aligned} if it is parallel to $k$ coordinate axes; similarly, a $k$-dimensional polytope ($k$-polytope) is \emph{axis-aligned} if its affine hull is an axis-aligned $k$-flat.
A $k$-flat or $k$-polytope is \emph{vertical} if its affine hull is parallel to the $x_d$-axis; or equivalently if it contains a nondegenerate line segment parallel to the $x_d$-axis.

\begin{figure}[htbp]
    \centering
	\includegraphics{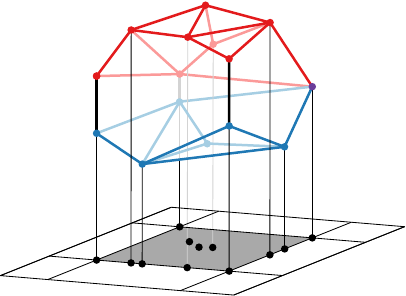}%
	\caption{A convex polyhedron $\conv(P)$ in $\mathbb{R}^3$, whose projection
	$\conv(P)^\proj$ is a rectangle. Seven points in $P$ are projected onto the four vertices  of $\conv(P)^\proj$. Overall, the silhouette of $P$ contains twelve points, five of which project into the relative interior of some edges of $\conv(P)^\proj$. Red (blue) vertices are upper (lower); the purple point is both upper and lower.}%
	\label{fig:silhouette}
\end{figure}

The proof of Lemma~\ref{lem:upperd} is constructive. For integers $0\leq i\leq j$,
we recursively define a $2^j\times \ldots \times 2^j\times 2^i$ grid in $\mathbb{R}^d$, and denote it by $S_d(i,j)$. For $i\geq 1$, the grid $S_d(i,j)$ is the disjoint union of two translated copies of $S_d(i-1,j)$, one above the other. Then we prove that every $d$-dimensional subset $P\subset S_d(i,j)$ in convex position contains at most $2^{d(d+1)}\cdot i\cdot j^{d-2}$ upper (resp., lower) vertices. For $i=j=\log n$, this implies that $|P|\leq 2\cdot 2^{d(d+1)}\cdot j^{d-1}=\Theta(\log^{d-1}n)$. The number of upper vertices in $P$ can be bounded as follows (the case of lower vertices is analogous). We partition the upper vertices in $P$ into three subsets: (1)~Upper vertices $p\in P$ such that $p^\proj$ is a vertex of $\conv(P)^\proj$ in $(d-1)$-dimensions. As $p^\proj$ may be an upper or a lower vertex in $S_{d-1}(j,j)$, by induction, there are at most $2\cdot 2^{(d-1)d}\cdot j^{d-2}$ such points. (2)~Upper vertices $p\in P$ such that $p^\proj$ lies in the interior of $\conv(P)^\proj$. We prove that one of the two translated copies of $S_d(i-1,j)$ contains these points as upper vertices, and by induction there are at most $2^{d(d+1)}\cdot (i-1)\cdot j^{d-2}$ such points.
(3)~Upper vertices $p\in P$ such that $p^\proj$ lies on the boundary of $\conv(P)^\proj$, but in the relative interior of some face of $\conv(P)^\proj$; see Figure~\ref{fig:silhouette} for examples. We show that only axis-aligned faces of $\conv(P)^\proj$ can contain interior grid points. We can control the number of such faces and such interior points by induction. Summation over axis-aligned faces of dimensions $1,\ldots , d-1$  yields an upper bound of $(2^{d-1}-2)2^{(d-1)d+3}\cdot j^{d-2}$ for such points. Summation over (1)--(3) yields the desired upper bound $2^{d(d+1)}\cdot i\cdot j^{d-2}$ for the number of upper (resp., lower) vertices.

\begin{proof}[Proof of Lemma~\ref{lem:upperd}]
It is enough to prove the claim when $n=2^j$ for integers $j\geq 0$. For integers $d\geq 3$ and $0\leq i\leq j$, we recursively define the grid $S_d(i,j)$ as a Cartesian product of $d$ sets of reals, where the first $d-1$ sets have $2^j$ elements and the last set has $2^i$ elements. In particular, $|S_d(i,j)|=2^{i+(d-1)j}$. We then show that $S_d(i,j)$ does not contain the vertex set of any full-dimensional convex polyhedron with more than $2^{(d^2+d+1)}\cdot i\cdot j^{d-2}$ vertices.

To initialize the recursion, we define boundary values as follows: For every integer $j\geq 0$, let $S_2(j,j)$ be the $2^j\times 2^j$ grid defined in the proof of Lemma~\ref{lem:upper2} that does not contain more than $4(j+1)$ points in convex position. Note that every line that contains 3 or more points from $S_2(j,j)$ is axis-parallel (this property was not needed in the proof of Lemma~\ref{lem:upper2}).
Assume now that $d\geq 3$, and $S_{d-1}(j,j)$ has been defined for all $j\geq 0$; and for all $k\in \{1,\ldots , d-2\}$, if a point $p\in S_{d-1}(j,j)$ lies in the relative interior of $k$-polytope $F$ whose vertices are in $S_{d-1}(j,j)$, then $F$ is axis-aligned.

Let $j$ be a nonnegative integer. We now construct $S_d(i,j)$ for all integers $0\leq i\leq j$ as follows.
Let $S_d(0,j)=S_{d-1}(j,j)\times \{0\}$. For $i=1,\ldots, j$, we define $S_d(i,j)$ as the disjoint union of two translates of $S_d(i-1,j)$. Specifically, let $S_d(i,j)=A\cup B$, where $A=S_d(i-1,j)$ and $B=A+\lambda_d^i\mathbf{e}_d$, where $\lambda_d^i>0$ is
a (large) scalar satisfying the following two conditions.
\begin{enumerate}[label={\rm (C\arabic*)},series=C]
\item\label{C1} For all $p\in A\cup B$ and $Q\subset A\cup B$, if $p$ is in the relative interior of $\conv(Q)$, then $\conv(Q)$ is an axis-aligned polytope.
\item\label{C2} For all $p_a\in A$, $p_b\in B$, and $Q\subset S_d(0,j)=S_{d-1}(j,j)\times \{0\}$, if both $p_a^\proj$ and $p_b^\proj$ are in the relative interior of $\conv(Q)$, then the line $p_ap_b$ intersects the horizontal hyperplane $x_d=0$ in the relative interior of $\conv(Q)$.
\end{enumerate}

We show that a scalar $\lambda_d^i$ satisfying both \ref{C1} and \ref{C2} exists. 
First, we may assume that $\lambda_d^i>0$ is sufficiently large so that $A$ and $B$ are separated by some horizontal hyperplane.
Consider a set $Q\subset A\cup B$ such that $Q$ contains points in both $A$ and $B$. Let $F=\conv(Q)$, $F_a=\conv(Q\cap A)$, and $F_b=\conv(Q\cap B)$. 
Since $\conv(Q)=\conv(F_a\cup F_b)$, we have $F\subset \conv(F_a\cup B)$ and $F\subset \conv(A\cup F_b)$. We set $\lambda_d^i>0$ sufficiently large so that every point in $A\cap \conv(F_a\cup B)$ is on a vertical line passing through a point in $F_a$, and every point in $B\cap \conv(A\cup F_b)$ is on a vertical line passing through a point in $F_b$. Assume that this condition holds for all $Q\subset A\cup B$, where $Q\cap A\neq \emptyset$ and $Q\cap B\neq \emptyset$.

We show that \ref{C1} holds if $\lambda_d^i$ is chosen as specified above. 
Let $Q\subset A\cup B$ and let $p\in A\cup B$ be in the relative interior of $\conv(Q)$.
Let $F=\conv(Q)$, $F_a=\conv(Q\cap A)$, and $F_b=\conv(Q\cap B)$, as above, and $F$ is a $k$-dimensional polytope for some $k\in \{1,\ldots ,d\}$. We distinguish between several cases based on the location of $p$.
(1) If $p$ is in the relative interior of $F_a$ (resp., $F_b$), then $F_a$ (resp., $F_b$) is axis-aligned by induction, and it has the same dimension as $F$, hence $F$ is also axis-aligned. 
(2) Otherwise $p$ is contained in the relative interior of neither $F_a$ nor $F_b$. By Carath\'eodory's theorem~\cite{eggleston1958}, $p$ is in the relative interior of some simplex $F'=\conv(Q')$ with vertex set $Q'\subset Q$ such that $Q'$ has vertices in both $A$ and $B$. By the choice of $\lambda_d^i$, point $p$ is on the vertical line passing through a point in $Q'\cap A$ or $Q'\cap B$. This implies that $p$ lies on a vertical line segment contained in $F'$, hence in $F$. Consequently, $F$ is parallel to the $x_d$-axis, and $F^\proj$ is a $(k-1)$-polytope. If $k=1$, this already proves that $F$ is axis-aligned. Assume that $k>1$ and note that $p^\proj\in S_{d-1}(j,j)$ and it is in the relative interior of $F^\proj$. By induction, $F^\proj$ is an axis-aligned $(k-1)$-polytope. Overall, $F$ is parallel to the $x_d$-axis and the $k-1$ coordinate axes parallel to $F^\proj$, consequently $F$ is axis-aligned, as required.

To satisfy condition \ref{C2}, notice that for any $p_a\in A$ and $p_b\in B$, the intersection of the line $p_a p_b$ and the hyperplane $x_d=0$ can be arbitrarily close to $p_a^\proj$ if $\lambda_d^i>0$ is sufficiently large.
This completes the definition of $S_d(i,j)$ for $d\geq 3$ and $0\leq i\leq j$.

Note, however, that in each iteration of the recursive contraction of $S_d(i,j)$, we only imposed lower bounds for the scalars $\lambda_d^i$, and different values of these parameters would yield different point sets. For all integers $d\geq 1$ and $0\leq i\leq j$, let $\mathcal{S}_d(i,j)$ denote the collection of all $2^j\times \ldots \times 2^j\times 2^i$ grids that can be obtained by the recursive construction above, in particular every grid in $\mathcal{S}_d(i,j)$ satisfies both \ref{C1} and \ref{C2} in each iteration. Our proof relies on the following relation between these constructions. 

\begin{claim}\label{cl:0}
Let $2\leq \ell\leq d$, let $S_d(i,j)\in \mathcal{S}_d(i,j)$, and let $W$ be an axis-aligned vertical $\ell$-flat such that $S_d(j,j)\cap W\neq \emptyset$. Then $S_d(i,j)\cap W\in \mathcal{S}_\ell(i,j)$.
\end{claim}
To prove Claim~\ref{cl:0}, note first that $S_d(i,j)\cap W$ is a $2^j\times \ldots \times 2^j\times 2^i$ grid. It can be obtained by the recursive construction above with suitable parameters $\lambda_d^i$ in the recursive steps. Indeed, assume that $S_d(i,j)=A\cup B$, where $A=S_d(i-1,j)$ and $B=A+\lambda_d^i\mathbf{e}_d$. Then $S_d(i,j)\cap W=(A\cap W)\cup (B\cup W)$. As both \ref{C1} and \ref{C2} are formulated in terms of the relative interior of a set $Q\subset A\cup B$, both conditions hold for $Q\subset (A\cup B)\cap W$, as well. This completes the proof of Claim~\ref{cl:0}.

In the remainder of the proof, we establish the following claim by induction:

\begin{claim}\label{cl:1}
Let $2\leq d$ and $S_d(i,j)\in \mathcal{S}_d(i,j)$. If $P\subset S_d(i,j)$ is a $d$-dimensional set in convex position, then $P$ contains at most $2^{d(d+1)}\cdot i\cdot j^{d-2}$ upper (resp., lower) vertices of $\conv(P)$.
\end{claim}
We prove Claim~\ref{cl:1} by induction on $d+i$. We focus on the number of upper vertices in $P$, the case of lower vertices follows analogously (e.g., by reflection in a horizontal hyperplane). 

In the base case, we have $d=2$ and $i=0$. Then the set $S_d(0,j)=S_{d-1}(j,j)\times \{0\}$ lies in a horizontal hyperplane in $\mathbb{R}^d$, and so a subset $P\subset S_d(0,j)$ cannot be full-dimensional, hence the claim vacuously holds for all $P\subset S_d(0,j)$. We observe a key property for $P$ (which will be used in the case $i=1$). If $d=2$ and $P\subset S_d(0,j)$ is a 1-dimensional set, then $P$ has precisely two extreme points in $\mathbb{R}^2$.
If $d\geq 3$ and $P\subset S_d(0,j)$ is a $(d-1)$-dimensional set, then by induction it contains at most $2^{(d-1)d}\cdot j \cdot j^{d-3}$ upper (resp., lower) vertices in $\mathbb{R}^{d-1}$, hence it has at most $2\cdot 2^{(d-1)d}\cdot j^{d-2}$ extreme points in $\mathbb{R}^d$.

Assume next that $i=1$, and either $d=2$ or the claim holds for $S_{d-1}(j,j)$.
We prove the claim for $S_d(1,j)$. The set $S_d(1,j)$ is the disjoint union of $A=S_d(0,j)$ and $B=S_{d}(0,j)+\lambda_d^1\mathbf{e}_d$. Let $P\subset S_d(1,j)$ be in convex position. Then every upper (resp., lower) vertex of $P$ is an extreme vertex in $P\cap A$ or $P\cap B$, hence $P$ contains at most
$2\cdot 2\cdot 2^{(d-1)d}\cdot j^{d-2}
=2^{d^2-d+2}\cdot j^{d-2}\leq 2^{d(d+1)}\cdot j^{d-2}
=2^{d(d+1)}\cdot 1\cdot j^{d-2}$
upper (resp., lower) vertices. 

In the general case, we assume that $2\leq i\leq j$, the claim holds for $S_d(i-1,j)$, and either $d=2$ or the claim holds for $S_{d-1}(j,j)$, as well. We prove the claim for $S_d(i,j)$. Recall that $S_d(i,j)$ is the disjoint union of two translates of $S_d(i-1,j)$, namely $A=S_d(i-1,j)$ and $B=S_d(i-1,j)+\lambda_d^i\mathbf{e}_d$. Let $P\subset S_d(i,j)$ be a full-dimensional set in convex position. We partition the upper vertices in $P$ as follows. For every upper vertex $p$, we have $p^\proj\in \conv(P)^\proj$,
and $p^\proj$ is either a vertex of $\conv(P)^\proj$ or lies in the relative interior of a $k$-face of $\conv(P)^\proj$ for some $k\in \{1,\ldots , d-1\}$.
Let $P_0$ be the set of upper vertices $p\in P$ such that $p^\proj$ is a vertex of $\conv(P)^\proj$.
For all $k\in \{1,\ldots ,d-1\}$, let $P_k\subset P$ be the set of upper vertices of $P$ such that $p^\proj$ lies in the relative interior of a $k$-face of $\conv(P)^\proj$.
Then $\bigcup_{k=0}^{d-1} P_k$ is the set of all upper vertices in $P$.
%

The orthogonal projection of $S_d(i,j)$ to the hyperplane $x_d=0$ is $S_{d-1}(j,j)$, and the orthogonal projection of $P_0$, denoted $P_0^\proj\subset S_{d-1}(j,j)$, is the vertex set of a $(d-1)$-dimensional convex polyhedron. By induction, $|P_0|\leq 2\cdot 2^{(d-1)d}\cdot j\cdot j^{d-3}=2^{(d-1)d+1}j^{d-2}$. To derive an upper bound on $|P_k|$ for $k\in \{1,\ldots , d-1\}$,
note that only axis-aligned faces of $\conv(P)^\proj$ can contain grid points of $S_{d-1}(j,j)$ in their relative interior by condition \ref{C1}. 
In order to complete the proof of Claim~\ref{cl:1}, we prove the following.

\begin{claim}\label{cl:2}
For every axis-aligned face $F$ of $\conv(P)^\proj$,
either $A$ or $B$ contains all upper (resp., all lower) 
vertices that project to the relative interior of $F$.
\end{claim}

Let $F$ be an axis-aligned $k$-face of $\conv(P)^\proj$ for $k\in \{1,2,\ldots , d-1\}$.
Let $P(F)\subset P$ be the set of upper vertices $p\in P$ such that $p^\proj$ lies in the relative interior of $F$,
and let $V(F)$ be the vertex set of $F$.
Since $F$ is a face of the convex polytope $\conv(P)^\proj$ in the horizontal hyperplane $x_d=0$,
we have $V(F)\subset S_{d-1}(j,j)$ and $V(F)$ is in convex position.
Consider the point set $P'=V(F)\cup P(F)$, and observe that if $P(F)\neq \emptyset$,  then $P'$ is the vertex set of the $(k+1)$-polytope $\conv(P')$, in which $F$ is one of the facets.

It remains to show that $P(F)\subseteq A$ or $P(F)\subseteq B$. Suppose, for the sake of contradiction, that $P(F)$ contains points from both $A$ and $B$. Let $p_a$ be a vertex in $P(F)\cap A$ with the maximum $x_d$-coordinate. By construction, every point in $B$ has higher $x_d$-coordinate than any point in $A$ (including $p_a$ and all vertices of $V(F)$). The 1-skeleton of $\conv(P')$ contains a $x_d$-monotonically increasing path from $p_a$ to an $x_d$-maximal vertex in $P'$. Let $p_b$ be the neighbor of $p_a$ along such a path. By the maximality of $p_a$, we have $p_b\in B$. Then $p_a p_b$ is an edge of $\conv(P')$, hence the line $p_a p_b$ is disjoint from the interior of $\conv(P')$. However, by condition \ref{C2}, the line $p_a p_b$ intersects the relative interior of facet $F$ of $\conv(P')$. This contradiction completes the proof of Claim~\ref{cl:2}.

We can now finish the proof of Claim~\ref{cl:1}. 
We have seen that $|P_0|\leq 2^{(d-1)d+1}j^{d-2}$.
It remains to bound $|P_k|$ for $k\in \{1,\ldots , d-1\}$. Note that $\conv(P)^\proj$ is a $(d-1)$-dimensional polytope. We can apply Claim~\ref{cl:2} for $k=d-1$ and $F=\conv(P)^\proj$, and conclude that either $A$ or $B$ contains all points in $P_{d-1}$. The points in $P_{d-1}$ are upper vertices of $\conv(P_{d-1})$. Since both $A$ and $B$ are translates of $S_d(i-1,j)$, the induction hypothesis yields  $|P_{d-1}|\leq 2^{d(d+1)}\cdot (i-1)\cdot j^{d-2}$.

For $k\in \{1,\ldots , d-2\}$, the axis-aligned $k$-faces of $\conv(P)^\proj$
are partitioned into $\binom{d-1}{k}$ equivalence classes. In each equivalence class, the $k$-faces $F$ are parallel to the same coordinate axes. In the projection to an orthogonal $(d-k-1)$-dimensional space (of the hyperplane $x_d=0$), they become vertices of a $(d-k-1)$-polytope. By induction, the number of (upper and lower) vertices of such a polytope is
bounded above by
$2\cdot 2^{(d-k-1)(d-k)}j^{(d-k-1)-1}=2^{(d-k-1)(d-k)+1}j^{d-k-2}$;
this is an upper bound on the size of the equivalence class.
Each axis-aligned $k$-face $F$ of $\conv(P)^\proj$ is the orthogonal projection of some axis-aligned $(k+1)$-face of $\conv(P)$, that we denote by $\overline{F}$. Claim~\ref{cl:0} with $\ell=k+1\leq d-1$ implies that the induction hypothesis holds for $\overline{F}$, and it yields an upper bound of $2^{(k+1)(k+2)}\cdot j^{k}$ for the number of upper vertices in $\overline{F}$. Overall, we have

\vspace{-\baselineskip}
\begin{align*}
|P_k|
&\leq \binom{d-1}{k}\cdot 2^{(d-k-1)(d-k)+1}j^{d-k-2}\cdot 2^{(k+1)(k+2)}\cdot j^{k}\\
&\leq \binom{d-1}{k}\cdot 2^{(d-1)d+3}\cdot j^{d-2},
\end{align*}
where we used that $(d-k-1)(d-k)+1+(k+1)(k+2)= (d-1)d+2k(k+2-d)+3\leq (d-1)d+3$ for all
$k\in \{1,\ldots , d-2\}$ and $d\geq 3$. Altogether, the number of upper vertices is

\vspace{-\baselineskip}
\begin{align*}
\sum_{k=0}^{d-1}|P_k|
&= |P_0|+\left(\sum_{k=1}^{d-2} |P_k|\right)+|P_{d-1}|\\
&\leq 2\cdot 2^{(d-1)d}\cdot j^{d-2}  +
    \left(\sum_{k=1}^{d-2}\binom{d-1}{k} 2^{(d-1)d+3}\cdot j^{d-2}\right)
    + 2^{d(d+1)}\cdot (i-1)\cdot j^{d-2}\\
&< 2^{d-1}\cdot 2^{(d-1)d+3}j^{d-2} + 2^{d(d+1)} \cdot (i-1)\cdot j^{d-2}\\
&< 2^{d(d+1)}j^{d-2} + 2^{d(d+1)} \cdot (i-1)\cdot j^{d-2}\\
&= 2^{d(d+1)}\cdot i \cdot j^{d-2},
\end{align*}
as required, where we used the binomial theorem.
This completes the proof of Claim~\ref{cl:1}.

Claim~\ref{cl:1} immediately implies Lemma~\ref{lem:upperd} by taking $n=2^j$ and $X=S_d(j,j)$, since $2\cdot 2^{d(d+1)}\cdot j^{d-1}=\Theta(\log^{d-1}n)$ for every fixed $d\geq 2$.
\end{proof}

\subsection{Lower bounds in higher dimensions}
\label{ssec:LowerBoundHighDim}

The proof technique in Section~\ref{ssec:lowerbound2D} is insufficient for establishing a lower bound of $\Omega(\log^{d-1}n)$
for $d\geq 3$. Whereas a $d$-dimensional $n\times \ldots \times n$ grid contains $\Omega(n^\delta)$ points in general position for some $\delta=\delta(d)>0$~\cite{cardinal2017}, the current best lower bound on the number of points in convex position in any set of $n$ points in general position in $\mathbb{R}^d$ is $\Omega(\log n)$; the conjectured value is $\Omega(\log^{d-1} n)$.
Instead, we rely on the structure of Cartesian products and induction on $d$. Our main result in this section is the following.

\begin{theorem}
Every $d$-dimensional Cartesian product $\Cross_{i=1}^d X_i$, where $|X_i|=n$ and $d$ is fixed, contains $\Omega(\log^{d-1}n)$ points in convex position.
\label{thm:lowerd}
\end{theorem}

We say that a strictly increasing sequence of real numbers $A=(a_1,\ldots, a_n)$ has the \emph{monotone differences} property (for short, $A$ is \emph{MD}) if
\begin{itemize}
\item $a_{i+1}-a_i > a_i-a_{i-1}$ for $i=2,\ldots, n-1$, or
\item $a_{i+1}-a_i < a_i-a_{i-1}$ for $i=2,\ldots, n-1$.
\end{itemize}
Furthermore, the sequence $A$ is \emph{$r$-MD} for some $r>1$ if
\begin{itemize}
\item $a_{i+1}-a_i \geq r(a_i-a_{i-1})$ for $i=2,\ldots, n-1$, or
\item $a_{i+1}-a_i \leq (a_i-a_{i-1})/r$ for $i=2,\ldots, n-1$.
\end{itemize}
A finite set $X\subseteq \mathbb{R}$ is \emph{MD} (resp., \emph{$r$-MD}) if its elements arranged in increasing order form an MD (resp., $r$-MD) sequence. These sequences are intimately related to convexity: a strictly increasing sequence $A=(a_1,\ldots, a_n)$ is MD if and only if there exists a monotone (increasing or decreasing) convex function $f:\mathbb{R}\rightarrow \mathbb{R}$ such that $a_i=f(i)$ for all $i=1,\ldots, n$.  MD sets have been studied in additive combinatorics~\cite{ENZ99,Hegy86,RSS17,SS11}.

We first show that every $n$-element set $X\subseteq \mathbb{R}$ contains an MD subset of size $\Omega(\log n)$, and this bound is the best possible (Lemma~\ref{lem:gap-sequence}). In contrast, every $n$-term arithmetic progression contains an MD subsequence of $\Theta(\sqrt{n})$ terms: for example $(0,\ldots, n-1)$ contains the subsequence $(i^2: i=0,\ldots, \lfloor \sqrt{n-1}\rfloor)$. We then show that for constant $d\geq 2$, the $d$-dimensional Cartesian product of $n$-element MD sets contains $\Theta(n^{d-1})$ points in convex position. The combination of these results immediately implies that every $n\times \ldots \times n$ Cartesian product in $\mathbb{R}^d$ contains $\Omega(\log^{d-1}n)$ points in convex position.

The following lemma gives a lower bound for MD sequences.
It is known that a monotone sequence of $n$ reals contains a 2-MD sequence (satisfying the so-called \emph{doubling differences condition}~\cite{Ros81}) of size $\Omega(\log n)$~\cite[Lemma~4.1]{BM14}; see also \cite{CFPSS14,ELIAS20131} for related recent results.

\begin{lemma}
\label{lem:gap-sequence}
Every set of $n$ real numbers contains an MD subset of size $\lfloor (\log n)/2\rfloor+1$.
For every $n\in \mathbb{N}$, there exists a set of $n$ real numbers in which the size of
every MD subset is at most $\lceil \log n\rceil +1$.
\end{lemma}

\begin{proof}
Let $X=(x_0,\ldots , x_{n-1})$ be a strictly increasing sequence. Assume, without loss of generality, that $n=2^\ell+1$ for some $\ell\in \mathbb{N}$. We construct a sequence of nested intervals
\[[a_0,b_0]\supset [a_1,b_1]\supset \ldots \supset [a_\ell,b_\ell]\]
such that the endpoints of the intervals are in $X$ and the lengths of the intervals decrease by factors of 2 or higher, that is, $b_i-a_i\leq (b_{i-1}-a_{i-1})/2$ for $i=1,\ldots, \ell$.

We start with the interval $[a_0,b_0]=[x_0,x_{n-1}]$; and for every $i=0,\ldots, \ell-1$, we divide $[a_i,b_i]$ into two intervals at the median, and recurse on the shorter interval; see Figure~\ref{fig:intervals}.

\begin{figure}[htbp]
\centering
	\includegraphics{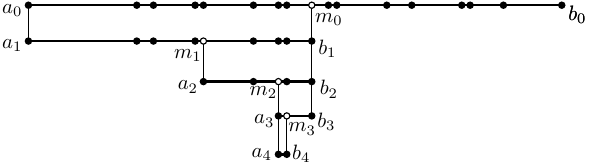}%
	\caption{A sequence $X$ of 17 elements and nested intervals $[a_0,b_0]\supset \ldots \supset [a_4,b_4]$.}
	\label{fig:intervals}
\end{figure}

By partitioning $[a_i,b_i]$ at the median, the algorithm maintains the invariant that $[a_i,b_i]$ contains $2^{\ell-i}+1$ elements of $X$. Note that for every $i=1,\ldots , \ell$, we have either ($a_{i-1}=a_i$ and $b_{i-1}>b_i$) or ($a_{i-1}<a_i$ and $b_i=b_{i-1}$). Consequently, the sequences $A=(a_0,a_1,\ldots , a_\ell)$ and $B=(b_\ell, b_{\ell-1},\ldots , b_0)$ both increase (not necessarily strictly), and at least one of them contains at least $1+\ell/2$ distinct terms. Assume, without loss of generality, that $A$ contains at least $1+\ell/2$ distinct terms. Let $C=(c_0,\ldots , c_k)$ be a maximal strictly increasing subsequence of $A$. Then $k\geq  \ell/2=\lfloor (\log n)/2\rfloor$.

We show that $C$ is an MD sequence. Let $i\in \{1,\ldots, k-1\}$.
Assume that $c_i=a_{j}=\ldots = a_{j'}$ for consecutive indices $j,\ldots ,j'$.
Then $c_{i-1}=a_{j-1}$, $c_i=a_j$, and $c_{i+1}=a_{j'+1}$. By construction, $c_i\in [a_{j-1},b_{j-1}]=[c_{i-1},b_j]$ such that $c_i-c_{i-1}\geq b_j-c_i$. Similarly, $c_{i+1}\in [a_{j'},b_{j'}]=[c_i,b_{j'}]$ such that
  $c_{i+1}-c_i\geq b_{j'}-c_{i+1}$. However, $[a_j,b_j]\subset [a_{j-1},b_{j-1}]$. As required, this yields
\[c_{i+1}-c_i = a_{j'+1}-a_j <b_j-a_j \leq \frac{b_{j-1}-a_{j-1}}{2}\leq a_j-a_{j-1}=c_i-c_{i-1}. \]

For the upper bound, consider the $n\times n$ grid $\{0,\ldots, n-1\}\times Y$, defined in the proof of Lemma~\ref{lem:upper2}, for which every chain in $\chainBottomLeft$ or $\chainBottomRight$ supported by
$\{0,\ldots, n-1\}\times Y$ has at most $\lceil \log n\rceil+1$ vertices.
We claim that the size of every MD subset of $Y\subset \mathbb{R}$ is at most $\lceil \log n\rceil +1$. Let $\{b_0,\ldots , b_{\ell-1}\}\subset Y$ be an MD subset such that $b_0<\ldots < b_{\ell-1}$. Then $\{(i,b_i): i=0,\ldots ,\ell-1\}\subset X\times Y$ is in $\chainBottomLeft$ or $\chainBottomRight$. Consequently, every MD subset of $Y$ has at most $\lceil \log n \rceil+1$ terms, as claimed.
\end{proof}

We show how to use Lemma~\ref{lem:gap-sequence} to establish a lower bound in the plane. While this approach yields worse constant coefficients than Lemma~\ref{lem:lowerGeneral}, its main advantage is that it generalizes to higher dimensions (see Lemma~\ref{lem:prod} below).

\begin{lemma}\label{lem:2D}
The Cartesian product of two MD sets, each of size $n$, supports $n$ points in convex position.
\end{lemma}
\begin{proof}
Let $A=\{a_1,\ldots, a_n\}$ and $B=\{b_1,\ldots ,b_n\}$ be MD sets
such that $a_i<a_{i+1}$ and $b_i<b_{i+1}$ for $i\in\{1,\ldots ,n-1\}$.
We may assume, by applying a reflection if necessary,
that $a_{i+1}-a_i < a_i-a_{i-1}$ and $b_{i+1}-b_i < b_i-b_{i-1},$
for $i\in \{2,\ldots, n-1\}$; see Figure~\ref{fig:lowerbound}.

\begin{figure}\centering
	\includegraphics{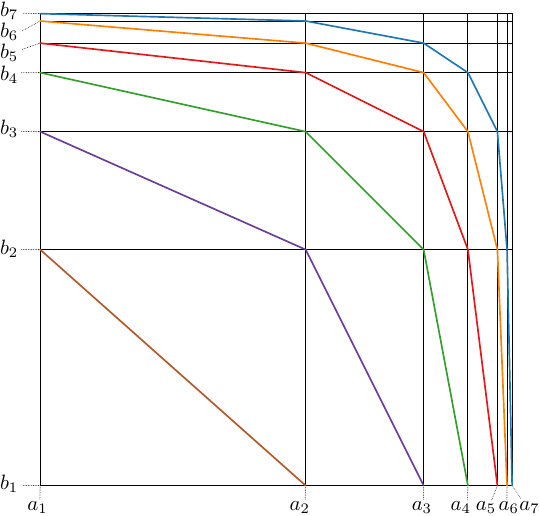}%
	\caption{A $7 \times 7$ grid $\{a_1,\ldots , a_7\}\times \{b_1,\ldots, b_7\}$,
where the differences between consecutive $x$-coordinates (resp., $y$-coordinates) decrease by factors of 2 or higher.
The point sets $\{(0,0)\}\cup \{(a_i,b_j): i+j=k\}$, for $k=2,\ldots, 8$, form nested convex chains.}%
	\label{fig:lowerbound}
\end{figure}

We define $P\subset A\times B$ as the set of $n$ points $(a_i,b_j)$ such that $i+j=n+1$.
By construction, every horizontal (vertical) line contains at most one point in $P$.
Since the differences $a_i-a_{i-1}$ are positive and strictly decrease in $i$;
and the differences $b_{n-i}-b_{n-i-1}$ are also positive and strictly decrease in $i$, the slopes $(b_{n-i}-b_{n-i-1})/(a_i-a_{i-1})$ strictly decrease, which proves the convexity of $P$.
\end{proof}

\begin{lemma}\label{lem:3D}
The Cartesian product of three MD sets, each of size $n$, contains $\binom{n+1}{2}$ points in convex position.
\end{lemma}

\begin{proof}
Let $A=\{a_1,\ldots,a_n\}$, $B=\{b_1,\ldots, b_n\}$, and $C=\{c_1,\ldots, c_n\}$
be MD sets, where the elements are labeled in increasing order.
We may assume, by applying a reflection in the $x$-, $y$-, or $z$-axis if necessary, that
\[a_{i+1}-a_i < a_i-a_{i-1}, \hspace{.3in}
  b_{i+1}-b_i < b_i-b_{i-1}, \hspace{.3in}
  c_{i+1}-c_i < c_i-c_{i-1},\]
for $i=2,\ldots, n-1$.
For $i,j,k\in \{1,\ldots, n\}$, let $p_{i,j,k}=(a_i,b_j,c_k)\in A\times B\times C$.
We can now let $P=\{p_{i,j,k}: i+j+k=n+2\}$. It is clear that $|P|=\sum_{i=1}^{n} i=\binom{n+1}{2}$.
We let $P'=P\cup \{p_{1,1,1}\}$ and show that the points in $P'$ are in convex position.

By Lemma~\ref{lem:2D}, the points in $P'$ lying in the planes $x=a_1$, $y=b_1$, and $z=c_1$ are each in convex position.
These convex $(n+1)$-gons are faces of the convex hull of $P$, denoted $\text{conv}(P)$. We show that the remaining faces of $\text{conv}(P)$ are the triangles $T'_{i,j,k}$ spanned by $p_{i,j,k}$, $p_{i,j+1,k-1}$, and $p_{i+1,j,k-1}$; and the triangles $T''_{i,j,k}$ spanned by $p_{i,j,k}$, $p_{i,j-1,k+1}$, and $p_{i-1,j,k+1}$.

The projection of these triangles to an $xy$-plane is shown in Figure~\ref{fig:lowerbound}. By construction, the union of these faces is homeomorphic to a sphere. It suffices to show that the dihedral angle between any two edge-adjacent triangles is convex.
Without loss of generality, consider triangle $T'_{i,j,k}$, which shares an edge with  (up to) three other triangles:
$T''_{i+1,j,k-1}$, $T''_{i,j+1,k-1}$, and $T''_{i+1,j+1,k-2}$. Consider first the triangles
$T'_{i,j,k}$ and $T''_{i+1,j+1,k-2}$. They share the edge $p_{i+1,j-1,k+1}p_{i-1,j+1,k+1}$, which lies in the $xy$-plane $z=c_{k+1}$. The orthogonal projections of these triangles to an $xy$-plane are congruent, however their extents in the $z$-axis
are $c_{i+1}-c_i$ and $c_i-c_{i-1}$, respectively. Since $c_{i+1}-c_i < c_i-c_{i-1}$, their dihedral angle is convex.
Similarly, the dihedral angles between $T'_{i,j,k}$ and $T''_{i+1,j,k-1}$ (resp.,  $T''_{i,j+1,k-1}$) is convex
because $a_{i+1}-a_i < a_i-a_{i-1}$ and $b_{i+1}-b_i < b_i-b_{i-1}$.
\end{proof}

The proof technique of Lemma~\ref{lem:3D} generalizes to higher dimensions:

\begin{lemma}\label{lem:prod}
For every constant $d\geq 2$, the Cartesian product of $d$ MD sets, each of size $n$, contains
$\Omega(n^{d-1})$ points in convex position.
\end{lemma}
\begin{proof}
We proceed by induction on $d$. For $d=2$ and $d=3$, Lemmas~\ref{lem:2D} and \ref{lem:3D} prove the claim.
Assume that $d\geq 4$, and the claim holds in lower dimensions.
For every $i\in \{1,\ldots, d\}$, let $A_i=\{a_{i,1},\ldots,a_{i,n}\}\subset \mathbb{R}$ be an MD set such that $a_{i,1}<\ldots<a_{i,n}$. We may assume, without loss of generality, that the differences between consecutive elements in $A_i$ strictly decrease for all $i\in\{1,\ldots , d\}$.

For every vector $\mathbf{v}=(v_1,\ldots , v_d)\in \{1,\ldots , n\}^d$, let $p_{\mathbf{v}}= (a_{1,v_1},a_{2,v_2},\ldots , a_{d,v_d})\in \Cross_{i=1}^d A_i$. Let $P=\{p_{\mathbf{v}}: \sum_{i=1}^d v_i=n+d-1\}$.
It is easy to see that $|P|=\Theta(n^{d-1})$. Let $P'=P\cup \{p_{1,\ldots ,1}\}$.
We show that the points in $P'$ are in convex position. We define a $(d-1)$-dimensional piecewise linear manifold $M$ as the union of $(d-1)$-dimensional convex polyhedra (facets). We show that $M$ is homeomorphic to the sphere $\mathbb{S}^{d-1}$. We also show that $M$ is the boundary of a convex polytope by verifying that the dihedral angle between any two adjacent facets is convex. It follows that $M$ is the boundary of $\conv(P')$, consequently $P'$ is in convex position.

\begin{figure}\centering
	\includegraphics{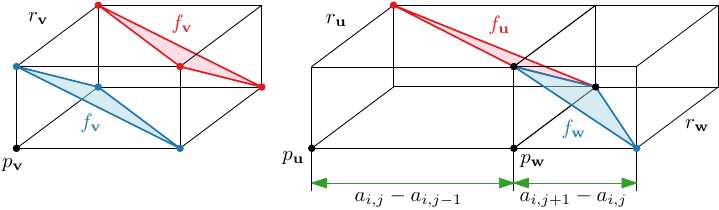}%
	\caption{Left: A rectangle $r_{\mathbf{v}}$ in $\mathbb{R}^3$, and two possible locations for $f_{\mathbf{v}}$, depending on the parity of $v_1+v_2+v_3$. (In general, the position of $f_{\mathbf{v}}$ in $r_{\mathbf{v}}$ depends on $\sum_{i=1}^d v_i\pmod{d-1}$.)
	Right: Two adjacent facets $f_{\mathbf{u}},f_{\mathbf{w}}\in \mathcal{F}$ in two adjacent hyperrectangles $r_{\mathbf{u}}$ and $r_{\mathbf{w}}$.}%
	\label{fig:tile}
\end{figure}

For every $i\in \{1,\ldots, d\}$, the points of $P'$ lying in the hyperplane $x_i=a_{i,1}$ are in convex position by induction, hence they each span a $(d-1)$-dimensional facets of $\text{conv}(P')$. For every vector $\mathbf{v}\in \{1,\ldots , n-1\}^d$, consider the hyperrectangle $r_{\mathbf{v}}= \Cross_{i=1}^d [a_{i,v_i},a_{i,v_i+1}]$, where $[a_{i,v_i},a_{i,v_i+1}]$ is the $i$-th \emph{extent} of $r_{\mathbf{v}}$. All vertices of each $r_{\mathbf{v}}$ are in $\Cross_{i=1}^d A_i$, and the hyperrectangles jointly tile the bounding box $\Cross_{i=1}^d [a_{i,1},a_{i,n}]$. For every $r_{\mathbf{v}}$, let $F_{\mathbf{v}}=\text{conv}(P\cap r_{\mathbf{v}})$, that is, the convex hull of vertices of $r_{\mathbf{v}}$ that are in $P$; see Figure~\ref{fig:tile}(left). By construction, every $F_{\mathbf{v}}$ is at most $(d-1)$-dimensional (possibly empty).
Let $\mathcal{F}$ be the set of $(d-1)$-dimensional polyhedra  $F_{\mathbf{v}}$, where $\mathbf{v}\in \{1,\ldots , n-1\}^d$, we call them \emph{facets}.
By construction, the union of the facets in $\mathcal{F}$, together with the facets in the hyperplanes $x_i=a_{i,1}$ for all $i\in\{1,\ldots, d\}$, is homeomorphic to  $\mathbb{S}^{d-1}$.
Consider two facets $f_{\mathbf{u}},f_{\mathbf{w}}\in \mathcal{F}$ that share a $(d-2)$-dimensional face; see Figure~\ref{fig:tile}(right).
Then $f_{\mathbf{u}}$ and $f_{\mathbf{w}}$ lie in two adjacent hyperrectangles, say $r_{\mathbf{u}}$ and $r_{\mathbf{w}}$, whose common boundary is $(d-1)$-dimensional,
and it is contained in a hyperplane $x_i=a_{i,j}$ for some $i\in \{1,\ldots , d\}$ and $j\in \{2,\ldots, n-1\}$.
The facet $f_{\mathbf{u}}$ (resp., $f_{\mathbf{w}}$) is parallel to the $(d-1$)-simplex spanned by the $d$ vertices of $r_{\mathbf{u}}$ (resp., $r_{\mathbf{w}}$) adjacent to $p_{\mathbf{u}}$ (resp., $p_{\mathbf{w}}$). Since $A_i$ is MD, we have $a_{i,j+1}-a_{i,j} < a_{i,j}-a_{i,j-1}$. Note that
$a_{i,j+1}-a_{i,j}$ and $a_{i,j}-a_{i,j-1}$ are the $i$-th extents of $r_{\mathbf{u}}$ and $r_{\mathbf{w}}$, respectively; and $r_{\mathbf{u}}$ and $r_{\mathbf{w}}$ have the same extent in the remaining $d-1$ coordinates. Consequently, the dihedral angle between $f_{\mathbf{u}}$ and $f_{\mathbf{w}}$ is convex, as required.
\end{proof}
Now Theorem~\ref{thm:lowerd} follows from Lemma~\ref{lem:gap-sequence} and Lemma~\ref{lem:prod}.

\section{Algorithms}\label{sec:algorithms}

In this section, we describe polynomial-time algorithms for
(i) finding convex chains and caps of maximum size; and
(ii) approximating the maximum size of a convex polygon;
where these structures are \emph{supported} by a given grid.
The main challenge is to ensure that the vertices of the convex polygon (resp., chain) have distinct $x$- and $y$-coordinates. The coordinates of a point $p \in X \times Y$ are denoted by $x(p)$ and $y(p)$.

As noted in Section~\ref{sec:intro}, efficient algorithms are available for finding a largest convex polygon or convex cap \emph{contained} in a planar point set. We briefly review these results as they provide an algorithm for the case of maximum chains supported by a grid, since the points in a convex chain have distinct $x$- and $y$-coordinates.

Given a set $P$ of $N$ points in the plane, Edelsbrunner and Guibas~\cite[Thm.~5.1.2]{EdelsbrunnerG89} use the dual line arrangement and dynamic programming to find the maximum size of a convex cap contained in \capTop\ in $O(N^2)$ time and $O(N)$ space; the same bounds hold for \capBottom, \capLeft, and \capRight. A convex cap of maximum length can be also computed in $O(N^2\log N)$ time and $O(N\log N)$ space.
Their method sweeps the dual line arrangement, in which each vertex (intersection of two dual lines) corresponds to a line segment $pq$ for $p,q\in P$, $p\neq q$. Specifically, for every pair of points $p,q\in P$, let $L(p,q)$ be the maximum length of a chain in \capTop\ whose \emph{last} two vertices are $p$ and $q$, respectively. Similarly, let $R(p,q)$ be the maximum length of a chain in \capTop\ whose \emph{first} two vertices are $p$ and $q$, respectively. The sweepline algorithm in~\cite{EdelsbrunnerG89} computes the values $L(p,q)$ (or all values $R(p,q)$) for all $p,q\in P$, $p\neq q$, in $O(N^2)$ time.

The following observation allows us to adapt the algorithm in~\cite{EdelsbrunnerG89} to find the maximum size of a convex cap in $\chainTopLeft$ (alternatively, to report a longest such chain) within the same time and space bounds. Analogous observations hold for $\chainTopRight$, $\chainBottomRight$, and $\chainBottomLeft$.

\begin{observation}\label{obs:monotone}
Let $P$ be a finite set of points in the plane, and let $p,q\in P$ such that $x(p)<x(q)$.
If $\slope(pq)>0$, then $L(p,q)$ is the maximum length of a chain in $\chainTopLeft$
whose last two vertices are $p$ and $q$.
\end{observation}

Consequently, the maximum length of a chain in $\chainTopLeft$ is $\max L(p,q)$,
where the maximum is taken over all pairs $p,q\in P$ where $x(p)<x(q)$ and $y(p)<y(q)$.
Since $x$- and $y$-coordinates do not repeat in a convex chain, we obtain the following result.

\begin{theorem}\label{thm:algo-chain}
In a given $n \times n$ grid, the maximum size of a supported convex chain can be computed in $O(n^4)$ time and~$O(n^2)$ space; and a supported convex chain of maximum size can be computed in $O(n^4\log n)$ time and~$O(n^2\log n)$ space.
\end{theorem}

\subsection{Convex caps}
\label{ssec:algoCap}	

In order to compute the maximum size of a convex cap in~$\capTop$, we must be careful to use each $y$-coordinate at most once.
We solve the more general problem of computing the maximum total size of two chains~$A\in\chainTopLeft$ and~$B\in\chainTopRight$ that use distinct~$y$-coordinates. Note that in contrast to the chains constituting a convex cap, the chains~$A$ and~$B$ may cross, or can have overlapping $x$-projections.

We present below a dynamic program to solve this problem, essentially building up solutions from bottom to top. To define the appropriate subproblem, we need to be able to refer to the last segment on the upper side of a chain. Hence, we refer to the last two vertices of $A$ as $l_1$ and $l_2$, setting~$l_2:=l_1$ if~$A$ is to consist of the single vertex~$l_1$. Similarly, we refer to the first two vertices of $B$ as $r_1$ and $r_2$, setting~$r_2:=r_1$ if~$B$ is to consist of the single vertex~$r_1$. The subproblem $C(l_1,l_2,r_1,r_2)$ then expresses the maximum total size of two chains $A \in \chainTopLeft$ and $B \in \chainTopRight$ with the given end vertices for $A$ and start vertices for $B$.
We claim that~$C(l_1,l_2,r_1,r_2)$ as defined below yields the desired quantity, or~$-\infty$ if no such chains~$A$ and~$B$ exist.
\[C(l_1,l_2,r_1,r_2) = \left\{\begin{array}{ll}
	-\infty & \text{if~$l_1\neq l_2$ and~$(l_1,l_2)\notin\chainTopLeft$, or}\\
    & \hspace{.6em}\text{$r_1\neq r_2$ and $(r_1,r_2)\notin\chainTopRight$, or}\\
	& \hspace{.5em}\text{$\{y(l_1),y(l_2)\}\cap\{y(r_1),y(r_2)\}\neq\emptyset$}\\
	2 & \text{else, if~$l_1=l_2$ and~$r_1=r_2$}\\
	L(l_1,l_2)+1
		& \text{else, if~$r_1=r_2$ and~$y(l_2)<y(r_1)$}\\
	R(r_1,r_2)+1
		& \text{else, if~$l_1=l_2$ and~$y(l_2)>y(r_1)$}\\
	\max\limits_{(v,l_1,l_2)\in\chainTopLeft \text{ or } v=l_1} C(v,l_1,r_1,r_2)+1
		& \text{else, if~$y(l_2)>y(r_1)$}\\
	\max\limits_{(r_1,r_2,v)\in\chainTopRight \text{ or } v=r_1} C(l_1,l_2,r_2,v)+1
		& \text{else,~$y(l_2)<y(r_1)$.}
\end{array}\right.\]

\begin{figure}[t]
  \centering
  \includegraphics[page=1]{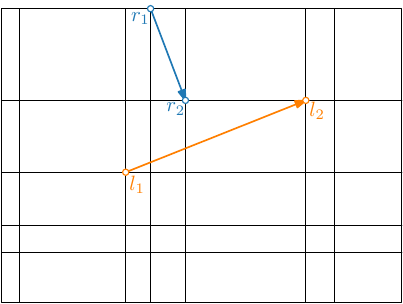}
  \includegraphics[page=3]{cap-dp}\\
  \hfill(a)\hfill(b)\hfill\hfill\hfill\\
  \includegraphics[page=4]{cap-dp}
  \includegraphics[page=5]{cap-dp}\\
  \hfill(c)\hfill(d)\hfill\hfill
  \caption{Illustration for the cases of $C(l_1,l_2,r_1,r_2)$.
  (a) Invalid configuration, as $y(r_2) = y(l_2)$.
  (b) $r_1=r_2$ is a single point above $l_2$, so we look for the longest chain ending in~$(l_1,l_2)$.
  (c) Removing the topmost point ($l_2$ in this case), testing all valid possible $v$ to find the longest chain. Note that the left and right chain may not complete to a cap -- this is checked separately.
  (d) We need test only whether $(l_1,l_2)$ and $(r_1,r_2)$ combine to make a cap (purple dotted line) to check whether the entry $(l_1,l_2,r_1,r_2)$ should be considered.}\label{fig:cap-dp}
\end{figure}

Informally, the first clause of the definition of~$C(l_1,l_2,r_1,r_2)$ gets rid of invalid arguments, for which no such chains $A\in\chainTopLeft$ and~$B\in\chainTopRight$ exist; see also Figure~\ref{fig:cap-dp}.
The second clause is the base case where $A$ and $B$ each have one vertex.
The third and fourth clause handle the case where the topmost chain consists of a single vertex---it can thus not interfere with the other (longer) chain anymore.
Here,~$L(p_1,p_2)$ (resp.,~$R(p_1,p_2)$) denotes the size of a largest convex chain $P$ in $\chainTopLeft$ (resp., $\chainTopRight$), ending (resp., starting) with vertices~$p_1$ and~$p_2$, or~$P=(p_1)$ if~$p_1=p_2$.
We can use the dynamic programming algorithm of~\cite{EdelsbrunnerG89} to compute~$L(p_1,p_2)$ and~$R(p_1,p_2)$.
The final two clauses deal with the case where the topmost chain consists of at least one edge, looking for the largest pair of chains that can be constructed by adding the topmost vertex to a `lower' pair of chains.
We prove the correctness of the formula in Lemma~\ref{lem:doubleChain}, using the following general observations.

\begin{observation}\label{obs:subsequence}
    If a supported convex polygon~$P$ is in a set~$\chainTopLeft$, $\chainTopRight$, $\chainBottomRight$, $\chainBottomLeft$, $\capTop$, $\capRight$, $\capBottom$, or~$\capLeft$, then every subsequence of~$P$ is in the same set.
    That is, these classes are hereditary.
\end{observation}

\begin{lemma}\label{lem:concatenate}
    Let $A = (a_1, \ldots, a_k)$ and $B = (b_1, \ldots, b_{k'})$ with $k, k' \geq 2$, $a_{k-1} = b_1$ and $a_k = b_2$.
    If $A$ and $B$ are both in the same set of $\chainTopRight$, $\chainTopLeft$, $\chainBottomRight$, or $\chainBottomLeft$, then their concatenation $D = (a_1, \ldots, a_k, b_3, \ldots, b_{k'})$ is in the same set.
\end{lemma}
\begin{proof}
Consider the case that $A$ and $B$ are in $\chainTopLeft$; the other cases are analogous. This means that both $x$-coordinates and their $y$-coordinates strictly increase in $A$ and $B$. As $a_k = b_2$, we know that $x(a_k) < x(b_i)$ and $y(a_k) < y(b_i)$ for all $i > 2$. Therefore, the concatenation $D = (a_1, \ldots, a_k, b_3, \ldots, b_{k'}$ has strictly increasing $x$- and $y$-coordinates. Except for its endpoints, each vertex of $D$ is an interior vertex of $A$ or $B$. Since $A$ and $B$ are convex, this readily implies that all interior vertices of $D$ are convex and hence $D$ is convex. Thus, $D$ is in $\chainTopLeft$ as well.
\end{proof}

\begin{lemma}\label{lem:doubleChain}
    Let~$k$ be the maximum total size of two chains~$A\in\chainTopLeft$ and~$B\in\chainTopRight$ that use distinct~$y$-coordinates, subject to the constraints that~$A$ ends in a given edge~$(l_1,l_2)$ (or~$A=(l_1)$ if~$l_1=l_2$) and~$B$ starts in a given edge~$(r_1,r_2)$ (or~$B=(r_1)$ if~$r_1=r_2$), or let~$k=-\infty$ if no such chains exists.
    Then~$C(l_1,l_2,r_1,r_2)=k$.
\end{lemma}
\begin{proof}
    Suppose that~$A\in\chainTopLeft$ and~$B\in\chainTopRight$ realize the maximum total size over all chains with distinct~$y$-coordinates under the given constraints.
    If the~$y$-coordinates of~$\{l_1,l_2\}$ and~$\{r_1,r_2\}$ are not disjoint, then~$A$ and~$B$ do not use distinct~$y$-coordinates, so~$C(l_1,l_2,r_1,r_2)$ correctly returns~$-\infty$.
    So suppose that the~$y$-coordinates differ.
    If~$l_1=l_2$ and~$r_1=r_2$, then~$A$ and~$B$ both have size~$1$, and we correctly return~$2$.

    If~$l_1\neq l_2$ and~$A\in\chainTopLeft$, then by Observation~\ref{obs:subsequence} we must have~$(l_1,l_2)\in\chainTopLeft$, so we correctly return~$-\infty$ if~$(l_1,l_2)\not\in\chainTopLeft$.
    Similarly, if~$r_1\neq r_2$ and~$(r_1,r_2)\not\in\chainTopRight$, we correctly return~$-\infty$.

    For the remaining cases, the point of~$\{l_1,l_2,r_1,r_2\}$ with maximum~$y$-coordinate will be either~$l_2$ or~$r_1$.
    If~$y(l_2)<y(r_1)$ and~$r_1=r_2$, then~$B$ has size~$1$ and~$A$ has only~$y$-coordinates at most that of~$l_2$.
    Since~$y(l_2)<y(r_1)$, any chain of~$\chainTopLeft$ that ends in~$(l_1,l_2)$ will have~$y$-coordinates distinct from~$B$, so~$A$ has the maximum size over all chains of~$\chainTopLeft$ ending in~$(l_1,l_2)$, as given by~$L(l_1,l_2)$. Therefore we correctly return~$L(l_1,l_2)+1$.
    Symmetrically, if~$y(r_1)<y(l_2)$ and~$l_1=r_2$, we correctly return~$R(r_1,r_2)+1$.

    Two (symmetric) cases remain: either~$y(l_2)<y(r_1)$ and~$r_1\neq r_2$, or~$y(r_1)<y(l_2)$ and~$l_1\neq l_2$.
    We use induction on the maximum~$y$-coordinate to prove that~$C(l_1,l_2,r_1,r_2)=k$ in these cases.
    We show the case where~$y(l_2)<y(r_1)$ and~$r_1\neq r_2$, the argument for the other case is symmetric.
    Removing the point~$r_1$ from chain~$B$ results in a chain~$B'\in\chainTopRight$ of size one less than that of~$B$.
    The chain~$B'$ starts in~$r_2$, and its~$y$-coordinates are still disjoint from those of~$A$.
    Since~$B$ had at least one edge, $B'$ consists either of a single vertex (namely~$r_2$), or it starts with an edge~$(r_2,v)$ for which~$(r_1,r_2,v)\in\chainTopRight$.
    Since these are exactly the terms we take the maximum over in the last clause of the definition of~$C(l_1,l_2,r_1,r_2)$, we by induction return at least~$k$.

    It remains to show that,~$C(l_1,l_2,r_1,r_2)\leq k$.
    For this, let~$A"\in\chainTopLeft$ and~$B"\in\chainTopRight$ be two chains that realize the maximum total size over all pairs of chains with distinct~$y$-coordinates, where~$A"$ ends in~$(l_1,l_2)$ and~$B"$ consists either of the single vertex~$r_2$, or it starts with an edge~$(r_2,v)$ for which~$(r_1,r_2,v)\in\chainTopRight$ (if no such chains exist, we correctly return~$-\infty+1=-\infty$).
    It suffices to show that adding~$r_1$ in front of chain~$B"$ yields a chain of~$\chainTopRight$ with~$y$-coordinates distinct from those of~$A"$.
    Indeed, since~$y(r_1)$ is greater than all other~$y$-coordinates, the $y$-coordinates remain distinct.
    Moreover, if~$B"$ consists of the single vertex~$r_2$, adding~$r_1$ in front of it yields~$(r_1,r_2)$, which we have already verified to lie in~$\chainTopRight$ (we would have returned~$-\infty$ otherwise).
    If~$B"\in\chainTopRight$ instead starts in an edge~$(r_2,v)$ for which~$(r_1,r_2,v)\in\chainTopRight$, then the chain with~$r_1$ in front also lie in~$\chainTopRight$, by Lemma~\ref{lem:concatenate}.
    Hence,~$C(l_1,l_2,r_1,r_2)$ returns at most $k$, and hence exactly~$k$.
\end{proof}

\begin{lemma}
    The value of $C$ for all inputs can be computed in~$O(n^{10})$ time and~$O(n^8)$ space.
\end{lemma}
\begin{proof}
    Deciding the applicable clause for an input takes constant time.
    Each clause takes at most~$O(n^2)$ time to compute, assuming that the referenced terms~$C$,~$L$, and~$R$ are already computed.
    For the third and fourth clauses, we can precompute all~$L(l_1,l_2)$ and~$R(r_1,r_2)$ within the desired time and space.
    The last two clauses of the equation reference~$O(n^2)$ inputs of~$C$ recursively, but with a smaller maximum~$y$-coordinate, so this recurrence is well-defined.
    There are~$O(n^8)$ possible inputs to~$C$, each of which takes constant additional space and~$O(n^2)$ time to compute, which gives the desired bounds.
\end{proof}

Any cap (of size at least~$2$) can be split into two chains~$A\in\chainTopLeft$ and~$B\in\chainTopRight$ with distinct~$y$-coordinates by removing the topmost edge.
Since unless~$n=1$, the largest cap has size at least~$2$, we can compute the size of the largest cap using~$C$ by taking the maximum over all inputs for whose concatenation (omitting~$l_1$ if~$l_1=l_2$ and~$r_2$ if~$r_1=r_2$) lies in~$\capTop$, as shown in Lemma~\ref{lem:cap}.

\begin{lemma}\label{lem:cap}
Suppose that~$A=(l_1,\dots,l_a)\in\chainTopLeft$ and~$B=(r_1,\ldots,r_b)\in\chainTopRight$ use distinct~$y$-coordinates, and let~$D$ be the concatenation of the last two vertices (or single vertex if~$A$ has size~$1$) of~$A$ and the first two vertices of~$B$ (or single vertex if~$B$ has size~$1$).
If~$D\in\capTop$, then the concatenation $D' = (l_1, \dots, l_a, r_1, \dots, r_b)$ of~$A$ and~$B$ lies in~$\capTop$.
\end{lemma}
\begin{proof}
Since $D\in\capTop$, we have $x(l_a)<x(r_1)$, so the $x$-coordinates of $D'$ are strictly increasing, and the $y$-coordinates are distinct.
Except $l_1$ and $r_b$, every vertex of $D'$ is a non-endpoint vertex of $A$, $B$ or $D$. Since these are convex, the vertex is convex in $D'$ as well; the entire concatenation is convex as a result. Thus, $D'$ lies in~$\capTop$.
\end{proof}

Testing if an input of~$C$ results in a cap takes constant time, so Theorem~\ref{thm:algo-cap} follows.
\begin{theorem}\label{thm:algo-cap}
For a given $n \times n$ grid, a supported convex cap of maximum size can be computed in $O(n^{10})$ time and $O(n^8)$ space.
\end{theorem}

\subsection{Convex \texorpdfstring{$n$}{n}-chains and \texorpdfstring{$n$}{n}-caps}
\label{ssec:ngons}

If we are solely interested in deciding whether the grid $X \times Y$, where $|X|=|Y|=n$, supports a convex chain or cap with precisely $n$ vertices, we can improve upon the previous algorithms considerably. Let~$X=\{x_1,\dots,x_n\}$ and~$Y=\{y_1,\dots,y_n\}$ with~$x_i<x_{i+1}$ and~$y_i<y_{i+1}$.
To test whether there is a chain of size $n$ in $\chainTopLeft$, it suffices to test whether the chain~$((x_1,y_1),\dots,(x_n,y_n))$ is in~$\chainTopLeft$, in $O(n)$ time.

To test whether there is a supported convex cap of size $n$ in~$\capTop$, we adapt the algorithm of Theorem~\ref{thm:algo-cap}.
Suppose~$P$ is a cap of~$\capTop$ of size~$n$, with $A_P$ and $B_P$ its maximal components in $\chainTopLeft$ and in $\chainTopRight$ respectively.
Then~$P$ uses all coordinates of~$X$, which restricts the types of chains~$A_P$ and~$B_P$ considerably.
In particular~$C(l_1,l_2,r_1,r_2)$ can be modified to consider only edges~$l$ and~$r$ that use consecutive~$x$-coordinates.

For~$1<k<n$ consider the subchains~$A_k\in\chainTopLeft$ and~$B_k\in\chainTopRight$ of~$A_P$ and~$B_P$ consisting only of vertices with~$y$-coordinate at most~$y_k$.
These chains have length~$k$ in total and use all of the coordinates~$\{y_1,\dots,y_k\}$.
Let~$(l_1,l_2)$ be the last edge of~$A_k$ and let~$(r_1,r_2)$ be the first edge of~$B_k$.
Then the coordinates~$y_{k-1}$ and~$y_k$ are used by~$l$, or by~$r$, or by~$l_2$ and~$r_1$.
Moreover, since the total length of~$A_k$ and~$B_k$ is~$k$, there are~$n-k$ unused~$X$-coordinates between~$l_2$ and~$r_1$, so if~$l_2.x=x_i$, then~$r_1.x=x_{i+n-k+1}$.
So for a fixed value of~$k$, we need only consider~$O(|Y|^2|X|)$ inputs for~$C(l_1,l_2,r_1,r_2)$.
Moreover, the recursive calls in the last two cases need only consider~$O(|Y|)$ values of~$v$.

This implies that there are $O(|Y|^3|X|)$ possible inputs to~$C(l_1,l_2,r_1,r_2)$ over all $k$.
As an entry now depends on $O(|Y|)$ subproblems and each is evaluated in constant time, the corresponding values can be computed in~$O(|Y|^4|X|)=O(n^5)$ time and~$O(n^4)$ space.
Similarly, we can test whether there is a cap of size $n$ in~$\capTop$ within the same time and space.

\subsection{Approximations} 

Although computing the maximum size of a supported convex polygon remains elusive,
we can easily devise a constant-factor approximation algorithm by eliminating duplicate coordinates as follows.
Compute a maximum size convex polygon~$P$ (possibly with duplicate coordinates) in a given $n\times n$ grid in~$O(n^6)$ time and $O(n^2)$ space~\cite[Thm.~5.1.3]{EdelsbrunnerG89}. Define a \emph{conflict graph} on the vertices of $P$, where two vertices are \emph{in conflict} if they share an $x$- or $y$-coordinate. Since each conflict corresponds to a horizontal or vertical line, the conflict graph has maximum degree at most 2 and contains no odd cycles, hence it is bipartite. One of the two sets in the bipartition contains at least half of the vertices of $P$ without duplicate coordinates, and so it determines a supported convex polygon. Since $P$ has $O(n)$ vertices, the conflict graph can be computed in $O(n)$ time. Overall, we obtain a $\frac12$-approximation for the maximum supported convex polygon in $O(n^6)$ time and $O(n^2)$ space.
The same strategy provides an $\frac12$-approximation for the maximum supported polygon in $\capTop$, $\capBottom$, $\capLeft$, and $\capRight$ in~$O(n^4)$ time and $O(n^2)$ space using~\cite[Thm.~5.1.2]{EdelsbrunnerG89}.
%


\section{The maximum number of convex polygons} \label{sec:counting}

Let $\Fn(n)$ be the maximum number of convex polygons that can be present in an $n \times n$ grid, with no restriction on the number of times each coordinate is used. Let $\Gn(n)$ be this number where all $2n$ grid lines are used (\ie, each grid line contains at least one vertex of the polygon). Let $\FSn(n)$ and $\GSn(n)$ be the corresponding numbers where each grid line is used at most once (so $\FSn(n)$ counts the maximum number of supported convex polygons). By definition, we have $\Fn(n) \geq \Gn(n) \geq \GSn(n)$ and $\Fn(n) \geq \FSn(n) \geq \GSn(n)$ for all $n \geq 2$. We prove the following theorem, in which the $\Theta^*(.)$ notation hides polynomial factors in $n$.

\begin{theorem}\label{thm:count}
 The following bounds hold:
\[
\Fn(n) = \Theta^*(16^n), \hspace{1cm}
\FSn(n) = \Theta^*(9^n), \hspace{1cm}
\Gn(n) = \Theta^*(9^n), \hspace{1cm}
\GSn(n) = \Theta^*(4^n).
\]
\end{theorem}

\subsection{Upper bounds}

We first prove that $\Fn(n) = O(n \cdot 16^n)$ by encoding each convex polygon in a unique way, so that the total number of convex polygons is bounded by the total number of encodings. Recall that a convex polygon $P$ can be decomposed into four convex chains $\chainTopLeft_P, \chainBottomLeft_P, \chainBottomRight_P, \chainTopRight_P$, with only extreme vertices of $P$ appearing in multiple chains. Let $\capLeft_P = \chainTopLeft_P \cup \chainBottomLeft_P$ and $\capBottom_P = \chainBottomLeft_P \cup \chainBottomRight_P$. To encode $P$, we assign the following number to each of the $2n$ grid lines $\ell$ (see Figure~\ref{fig:f5} for an example): 0 if $\ell$ is not incident on any vertex of $P$, 3 if $\ell$ is incident on multiple vertices of $P$, 1 if $\ell$ is incident on one vertex of $P$ and that vertex lies on $\capLeft_P$ if $\ell$ is horizontal, or on $\capBottom_P$ if $\ell$ is vertical, and 2 otherwise.

\begin{figure}[b]
\centering
\includegraphics{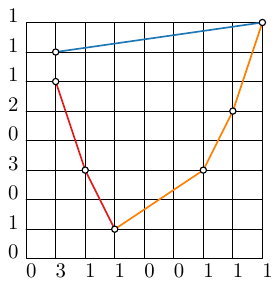}
\caption{Encoding the grid lines.}
\label{fig:f5}
\end{figure}

In addition, we record the index of the horizontal line containing the leftmost vertex of $P$ (pick the topmost of these if there are multiple leftmost points).

Since each of the $2n$ grid lines is assigned one of $4$ possible values, and there are $n$ horizontal lines, the total number of encodings is $O(n \cdot 4^{2n}) = O(n \cdot 16^n)$. All that is left to show is that each encoding corresponds to at most one convex polygon.

First, observe that if $P$ is a convex chain, say in $\chainTopLeft$, then the set of grid lines containing a vertex of $P$ uniquely defines $P$: since both coordinates change monotonically, the $i$-th vertex of $P$ must be the intersection of the $i$-th horizontal and vertical lines. So all we need to do to reconstruct $P$ is to identify the set of lines that make up each convex chain.

Since we know the location of the (topmost) leftmost vertex of $P$, we know where $\chainTopLeft_P$ starts. Every horizontal line above this point labelled with a $1$ or $3$ must contain a vertex of $\chainTopLeft_P$; let $k$ be the number of such lines. Since the $x$-coordinates are monotonic as well, $\chainTopLeft_P$ ends at the $k$-th vertical line labelled with a $2$ or $3$. The next chain, $\chainTopRight_P$, starts either at the end of $\chainTopLeft_P$, if the horizontal line is labelled with a $1$, or at the intersection of this horizontal line with the next vertical line labelled with a $2$ or $3$, if this horizontal line is labelled with a $3$. We can find the rest of the chains in a similar way. Thus, $\Fn(n) = O(n \cdot 16^n)$.

The upper bounds for $\FSn(n)$, $\Gn(n)$, and $\GSn(n)$ are analogous, except that certain labels are excluded. For the number of supported convex polygons $\FSn(n)$, each grid line is used at most once, which means that the label $3$ cannot be used. Thus, $\FSn(n) = O(n \cdot 3^{2n}) = O(n \cdot 9^n)$. Similarly, for $\Gn(n)$, all grid lines contain at least one vertex of the polygon, so the label $0$ cannot be used. Therefore $\Gn(n) = O(n \cdot 3^{2n}) = O(n \cdot 9^n)$. Finally, for $\GSn(n)$, every grid line contains exactly one vertex of the polygon, so neither $0$ nor $3$ can be used as labels. This gives $\GSn(n) = O(n \cdot 2^{2n}) = O(n \cdot 4^n)$ possibilities.

\subsection{Lower bounds}
\label{ssec:counting-lb}

\begin{figure}[b]
\centering
\includegraphics{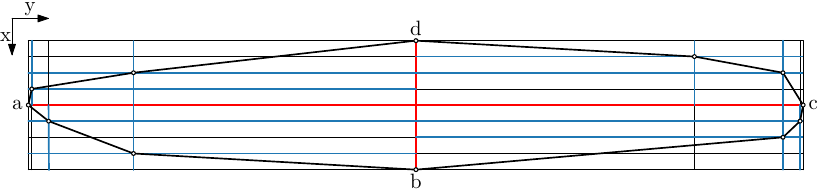}
\caption{
The $n\times n$ grid defined in Section~\ref{ssec:counting-lb},
with $n=9 =2m+3$ for $m=3$, before doubling the median lines (red).
The segments (parts of grid lines) incident to vertices are drawn in blue.
}
\label{fig:f1}
\end{figure}%

Assume that $n = 2m+4$, where $m \in \mathbb{N}$ satisfies suitable divisibility conditions, as needed.
All four lower bounds use the same grid, constructed as follows (see Figure~\ref{fig:f1}).
\begin{align*}
 X~~  &= \{1, \dots, n-1\} & Y^- &= \{y_1, \dots, y_{m+2}\}\text{, where }y_i = n^i\\
 Y~~  &= Y^- \cup Y^+      & Y^+ &= \{z_1, \dots, z_{m+2}\}\text{, where }z_i = 2 \cdot y_{m+2} - y_i
\end{align*}
Note that this results in an $(n-1) \times (n-1)$ grid, since $y_{m+2} = z_{m+2}$. To obtain an $n \times n$ grid, we duplicate the median grid lines in both directions and offset them by a sufficiently small  distance~$\varepsilon>0$. The resulting grid has the property
that any three points $p$, $q$, $r$ in the lower half $X \times Y^-$ with $x(p) < x(q) < x(r)$ and $y(p) < y(q) < y(r)$ make a left turn at $q$. To see this, suppose that $y(p) = n^i$, $y(q) = n^j$, and $y(r) = n^k$, for some $1 \leq i < j < k \leq n$. Then the slope of $
pq$ is strictly smaller than the slope of $qr$, since
\[
 \slope(qr) = \frac{n^k - n^j}{x(r) - x(q)} \geq \frac{n^{j+1} - n^j}{n - 1} = n^j >
 \frac{n^j - n^i}{1} \geq \frac{n^j - n^i}{x(q) - x(p)} = \slope(pq).
\]
Thus, any sequence of points with increasing $x$- and $y$-coordinates in the lower half is in~$\chainBottomRight$. By symmetry, such a sequence in the upper half $X \times Y^+$ is in~$\chainTopLeft$. Analogously, points with increasing $x$-coordinates and decreasing $y$-coordinates are in~$\chainBottomLeft$ if they are in the lower half and $\chainTopRight$ if they are in the upper half.

We first derive lower bounds on $\GSn(n)$ and $\Gn(n)$ by constructing a large set of convex polygons that use each grid line at least once. Then we use these bounds to derive the bounds on $\FSn(n)$ and $\Fn(n)$. The polygons we construct all share the same four extreme vertices, which lie on the intersections of the grid boundary with the duplicated median grid lines. Specifically, the leftmost and rightmost vertices are the intersections of the duplicate horizontal medians with the left and right boundary, and the highest and lowest vertices are the intersections of the duplicate vertical medians with the top and bottom boundary. Since each of these median lines now contain a vertex, we can choose additional vertices from the remaining $2m$ grid lines in each direction.

To construct each polygon, select $m/2$ vertical grid lines left of the median to participate in the bottom chain, and do the same right of the median. Likewise, select $m/2$ horizontal grid lines above and below the median, respectively, to participate in the left chain. The remaining grid lines participate in the other chain (top or right). This results in a polygon with $m/2$ vertices in each quadrant of the grid (excluding the extreme vertices). The convexity follows from our earlier observations. The total number of such polygons is

\[
 \binom{m}{\frac{m}{2}}^4 = \Theta\left( \left( m^{-\frac{1}{2}} 2^m \right)^4 \right)=
\Theta(m^{-2} 2^{4m})= \Theta(n^{-2} 2^{2n})=
\Theta(n^{-2} 4^n) = \Theta^*(4^n).
\]
The first step uses the following estimate, which can be derived from Stirling's formula for the factorial~\cite{dumitrescu2013bounds}. Let $0 < \alpha < 1$, then

\[
 \binom{n}{\alpha n}= \Theta(n^{-\frac{1}{2}} 2^{H(\alpha)n})\text{, where }H(\alpha)=-\alpha \log_2 \alpha - (1-\alpha) \log_2 (1-\alpha).
\]

For the lower bound on $\Gn(n)$, the only difference is that we now allow grid lines to contain vertices in two chains. We obtain a maximum when we divide the grid lines evenly between the three groups (bottom chain, top chain, both chains). Thus, we select $m/3$ vertical grid lines left of the median to participate in the bottom chain, another $m/3$ to participate in the top chain and the remaining $m/3$ participate in both. We repeat this selection to the right of the median and on both sides of the median horizontal line. As before, this results in a convex polygon with the same number of vertices in each quadrant of the grid---exactly $2m/3$ this time. The number of such polygons is

\vspace{-\baselineskip}
\begin{align*}
 \binom{m}{\frac{m}{3}}^4\binom{\frac{2m}{3}}{\frac{m}{3}}^4
 &= \Theta\left( \left( m^{-\frac{1}{2}} 2^{H(\frac{1}{3})m} \cdot m^{-\frac{1}{2}} 2^{H(\frac{1}{2})\frac{2m}{3}} \right)^4 \right)\\
 &= \Theta\left(m^{-4} 2^{4m(\log_2 3 - \frac{2}{3} + \frac{2}{3})}\right)
 = \Theta\left(n^{-4} 2^{2n\log_2 3}\right)
 = \Theta\left(n^{-4} 9^{n}\right)
 = \Theta^*\left(9^{n}\right).
\end{align*}

To translate these bounds to bounds on $\FSn(n)$ and $\Fn(n)$, where some grid lines may contain no vertices of the polygon, we observe that the arguments for the bounds above also work for a subgrid of $X \times Y$, provided that the subgrid includes the boundary and medians and has the same number of grid lines on each side of the median in both directions. For $\FSn(n)$, we select $2m/3$ grid lines on each side of each median (balancing the number of vertices with the two different chains) to make up our subgrid and plug in the bound on $\GSn(n)$, which yields

\[
 \binom{m}{\frac{2m}{3}}^4 \Omega^*(4^{\frac{2n}{3}})
 = \Omega^*(2^{2n(H(\frac{2}{3}) + \frac{2}{3})})
 = \Omega^*(2^{2n\log_2 3})
 = \Omega^*(9^n).
\]
Finally, for the bound on $\Fn(n)$, we select $3m/4$ grid lines on each side of each median (balancing the number of vertices with the three different options for a grid line in the proof of $\Gn(n)$), to make up our subgrid and plug in the bound on $\Gn(n)$, giving

\[
 \binom{m}{\frac{3m}{4}}^4 \Omega^*(9^{\frac{3n}{4}})
 = \Omega^*(2^{2n(H(\frac{3}{4}) + \frac{3}{4} \log_2 3)})
 = \Omega^*(2^{4n})
 = \Omega^*(16^n).
\]

\subsection{The maximum number of weakly convex polygons} \label{sec:maximum-weakly}

Let $W(n)$ denote the maximum number of weakly convex polygons that contained in an $n \times n$ grid. A polygon $P$ in $\mathbb{R}^2$ is \emph{weakly convex} if all of its internal angles are less than or equal to $\pi$. Here we identify each polygon by its set of boundary vertices, so different polygons may have identical convex hulls.
Since $W(n) \geq F(n)$, we have $W(n) = \Omega^*(16^n)$.
In fact, a slightly better lower bound trivially holds even for the $n\times n$
section of the integer lattice $Z_0=[n]\times [n]$. Consider all polygons
whose vertices are the four extreme vertices of $Z_0$, and an arbitrary subset of
the remaining $4n-8$ grid points in $\partial \text{conv}(Z_0)$. There are
$2^{4n-8}=\Omega(16^n)$ such subsets, and the lower bound $W(n) =\Omega(16^n)$ follows.

To show that $W(n) =O^*(16^n)$, we modify the previous encoding used to show
that $F(n) =O^*(16^n)$. While the four grid lines along the boundary of the bounding box of $P$
can be incident to arbitrarily many vertices, we still use at most two vertices
for each such line, namely at most two extreme vertices.
For each weakly convex polygon $P$, record the at most $8$ extreme vertices
incident to $\partial B$ together with a vector (sequence) of length $2n$:
$n$ elements corresponding to the horizontal lines (from the lowest to the highest),
and $n$ elements corresponding to the vertical lines (from left to right).
As previously, we encode each grid line by an element of $\{0,1,2,3\}$,
where $3$ stands for a line incident to at least two vertices.
By (weak) convexity, a grid line can be incident to $3$ or more vertices of $P$
only if it is one of the four lines along the bounding box of $P$.

From the recorded information, we can reconstruct a weakly convex polygon in the $n \times n$ grid.
Consequently, the number of convex polygons in the grid is bounded from above
by the number of encodings, namely $W(n)= O(n^{8} \cdot 4^{2n}) = O(n^{8} \cdot 16^n)=O^*(16^n)$.
We summarize the bounds we have obtained in the following.

\begin{theorem} \label{thm:W}
Let $W(n)$ denote the maximum number of weakly convex polygons that can be present in an
$n \times n$ grid. Then $W(n)= \Omega(16^n)$ and $W(n)= O^*(16^n)$.
\end{theorem}

\section{Conclusions}
\label{sec:con}

We studied combinatorial properties of convex polygons (resp., polytopes) in Cartesian products in $d$-space. Similar questions for point sets in general position or for lattice polygons (resp., polytopes) have been previously considered. We showed that every $n\times \ldots \times n$ Cartesian product in $\mathbb{R}^d$ contains $\Omega(\log^{d-1}n)$ points in convex position, and this bound is the best possible. Our upper bound matches previous bounds~\cite{KV03,Valtr92} for points in general position, which are conjectured to be tight. Our lower bound, however, does not yield any improvement for points in general position. In contrast, an $n\times \ldots \times n$ section of the integer lattice $\mathbb{Z}^d$ contains significantly more, namely $\Theta(n^{d(d-1)/(d+1)})$, points in convex position~\cite{And63}.

The \emph{maximum number} of convex polygons in an $n\times n$ Cartesian product is $F(n)=\Theta^*(16^n)$. This bound is tight up to polynomial factors, and is significantly larger than the corresponding bound in an $n\times n$ section of the integer lattice~\cite{barany1992number1}. In contrast, $n^2$ points in convex (hence general) position trivially determine $2^{n^2}-1$ convex polygons. Erd\H{o}s~\cite{erdos1978} proved that the \emph{minimum number} of convex polygons determined by $n$ points in general position is $\exp(\Theta(\log^2 n))$. Determining (or estimating) the minimum number of convex polygons in an $n\times n$ Cartesian product and in higher dimensions remain as open problems.

Our motivating problem was the reconstruction of a convex polygon from the $x$- and $y$-projections of its vertices. We presented a $\frac12$-approximation for computing the maximal size of a convex polygon supported by a grid $X\times Y$. Finding an efficient algorithm for the original problem, or proving its hardness, remains open. As our dynamic program does not directly extend to $d \geq 3$, approximation algorithms in higher dimensions are also of interest.

\paragraph{Acknowledgments.} We are grateful to the anonymous referees for their careful reading of the paper that helped clarify several subtle details in the inductive proof in Section~\ref{sec:bounds}.


\bibliography{gridgons}

\end{document}